\documentclass[pra,notitlepage,twocolumn,nofootinbib,superscriptaddress,longbibliography]{revtex4-2}

\usepackage{algorithm,algorithmic,amsfonts,amsmath,amssymb,mathtools,mathrsfs,bbm,float}
\usepackage{graphicx,epic,eepic,epsfig,latexsym,verbatim,color}
\usepackage[inline,shortlabels]{enumitem}
\usepackage[caption=false]{subfig}
 
\usepackage{tikz}
\usetikzlibrary{chains}
\usetikzlibrary{fit}

\usepackage{epsfig}
\usetikzlibrary{shapes.symbols,patterns} 
\usepackage{pgfplots}

\usepackage[strict]{changepage}

\usepackage[marginal]{footmisc}
\usepackage{url}
\usepackage{theorem}
\usepackage{thmtools}
\usepackage[ISO]{diffcoeff}
\usepackage{nicematrix}
\usepackage{xfrac}

\newtheorem{proposition}{Proposition}
\newtheorem{lemma}[proposition]{Lemma}

\newtheorem{theorem}[proposition]{Theorem}
\newtheorem{corollary}[proposition]{Corollary}

\newenvironment{proof}{\noindent \textbf{{Proof~} }}{\hfill $\blacksquare$}

\def\squareforqed{\hbox{\rlap{$\sqcap$}$\sqcup$}}
\def\qed{\ifmmode\squareforqed\else{\unskip\nobreak\hfil
\penalty50\hskip1em\null\nobreak\hfil\squareforqed
\parfillskip=0pt\finalhyphendemerits=0\endgraf}\fi}
\def\endenv{\ifmmode\;\else{\unskip\nobreak\hfil
\penalty50\hskip1em\null\nobreak\hfil\;
\parfillskip=0pt\finalhyphendemerits=0\endgraf}\fi}

\newcounter{example}

\mathchardef\ordinarycolon\mathcode`\:
\mathcode`\:=\string"8000
\def\vcentcolon{\mathrel{\mathop\ordinarycolon}}
\begingroup \catcode`\:=\active
  \lowercase{\endgroup
  \let :\vcentcolon
  }

\usepackage{hyperref}
\hypersetup{colorlinks=true,citecolor=blue,linkcolor=blue,filecolor=blue,urlcolor=blue,breaklinks=true}
\usepackage{footnotebackref}
\usepackage{cleveref}
\usepackage{graphicx}
\usepackage{xcolor}

\definecolor{darkblue}{RGB}{0,76,156}
\definecolor{darkkblue}{RGB}{0,0,153}
\definecolor{blue2}{RGB}{102,178,255}
\definecolor{darkred}{RGB}{195,0,0}

\RequirePackage[framemethod=default]{mdframed}
\newmdenv[skipabove=7pt,
skipbelow=7pt,
backgroundcolor=darkblue!15,
innerleftmargin=5pt,
innerrightmargin=5pt,
innertopmargin=5pt,
leftmargin=0cm,
rightmargin=0cm,
innerbottommargin=5pt,
linewidth=1pt]{tBox}

\newmdenv[skipabove=7pt,
skipbelow=7pt,
backgroundcolor=blue2!25,
innerleftmargin=5pt,
innerrightmargin=5pt,
innertopmargin=5pt,
leftmargin=0cm,
rightmargin=0cm,
innerbottommargin=5pt,
linewidth=1pt]{dBox}

\newmdenv[skipabove=7pt,
skipbelow=7pt,
backgroundcolor=darkred!15,
innerleftmargin=5pt,
innerrightmargin=5pt,
innertopmargin=5pt,
leftmargin=0cm,
rightmargin=0cm,
innerbottommargin=5pt,
linewidth=1pt]{rBox}

\newcommand{\nc}{\newcommand}
\nc{\ketbra}[2]{\lvert#1\rangle\!\langle#2\rvert}

\DeclarePairedDelimiter{\norm}{\lVert}{\rVert}
\DeclarePairedDelimiter{\abs}{\lvert}{\rvert}
\DeclarePairedDelimiter{\ceil}{\lceil}{\rceil}

\makeatletter
\let\oldabs\abs
\def\abs{\@ifstar{\oldabs}{\oldabs*}}
\let\oldnorm\norm
\def\norm{\@ifstar{\oldnorm}{\oldnorm*}}
\makeatother

\DeclarePairedDelimiterX{\infdivx}[2]{(}{)}{%
  #1\;\delimsize\|\;#2%
}
\newcommand{\infdiv}{D\infdivx}

\nc{\proj}[1]{| #1\rangle\!\langle #1 |}
\nc{\avg}[1]{\langle#1\rangle}
\nc{\smfrac}[2]{\mbox{$\frac{#1}{#2}$}}
\nc{\tr}{\operatorname{tr}}
\nc{\ox}{\otimes}
\nc{\dg}{\dagger}
\nc{\dn}{\downarrow}
\nc{\cA}{{\cal A}}
\nc{\cB}{{\cal B}}
\nc{\cC}{{\cal C}}
\nc{\cD}{{\cal D}}
\nc{\cE}{{\cal E}}
\nc{\cF}{{\cal F}}
\nc{\cG}{{\cal G}}
\nc{\cH}{{\cal H}}
\nc{\cI}{{\cal I}}
\nc{\cJ}{{\cal J}}
\nc{\cK}{{\cal K}}
\nc{\cL}{{\cal L}}
\nc{\cM}{{\cal M}}
\nc{\cN}{{\cal N}}
\nc{\cO}{{\cal O}}
\nc{\cP}{{\cal P}}
\nc{\cQ}{{\cal Q}}
\nc{\cR}{{\cal R}}
\nc{\cS}{{\cal S}}
\nc{\cT}{{\cal T}}
\nc{\cU}{{\cal U}}
\nc{\cV}{{\cal V}}
\nc{\cX}{{\cal X}}
\nc{\cY}{{\cal Y}}
\nc{\cZ}{{\cal Z}}
\nc{\cW}{{\cal W}}
\nc{\csupp}{{\operatorname{csupp}}}
\nc{\qsupp}{{\operatorname{qsupp}}}
\nc{\var}{{\operatorname{var}}}
\nc{\rar}{\rightarrow}
\nc{\lrar}{\longrightarrow}
\nc{\polylog}{{\operatorname{polylog}}}
\nc{\wt}{{\operatorname{wt}}}
\nc{\supp}{{\operatorname{supp}}}

\nc{\argmin}{{\operatorname{argmin}}}

\makeatletter
\newcommand{\tpmod}[1]{{\@displayfalse\pmod{#1}}}
\makeatother

\def\a{\alpha}

\def\d{\delta}

\def\z{\zeta}

\def\x{\xi}

\def\c{\chi}

\def\L{\Lambda}

\nc{\RR}{{{\mathbb R}}}
\nc{\CC}{{{\mathbb C}}}
\nc{\FF}{{{\mathbb F}}}
\nc{\NN}{{{\mathbb N}}}
\nc{\ZZ}{{{\mathbb Z}}}
\nc{\PP}{{{\mathbb P}}}
\nc{\QQ}{{{\mathbb Q}}}
\nc{\UU}{{{\mathbb U}}}
\nc{\EE}{{{\mathbb E}}}
\nc{\id}{{\operatorname{id}}}

\nc{\CHSH}{{\operatorname{CHSH}}}

\nc{\rU}{\mbox{U}}

\nc{\ob}[1]{#1}

\nc{\SEP}{{\text{\rm SEP}}}
\nc{\NS}{{\text{\rm NS}}}
\nc{\LOCC}{{\text{\rm LOCC}}}
\nc{\PPT}{{\text{\rm PPT}}}
\nc{\EXT}{{\text{\rm EXT}}}
\nc{\Sym}{{\operatorname{Sym}}}

\nc{\ERLO}{{E_{\text{r,LO}}}}
\nc{\ERLOCC}{{E_{\text{r,LOCC}}}}
\nc{\ERPPT}{{E_{\text{r,PPT}}}}
\nc{\ERLOCCinfty}{{E^{\infty}_{\text{r,LOCC}}}}
\nc{\Aram}{{\operatorname{\sf A}}}

\newcommand{\eps}{\varepsilon}

\makeatletter
\def\grd@save@target#1{%
  \def\grd@target{#1}}
\def\grd@save@start#1{%
  \def\grd@start{#1}}
\tikzset{
  grid with coordinates/.style={
    to path={%
      \pgfextra{%
        \edef\grd@@target{(\tikztotarget)}%
        \tikz@scan@one@point\grd@save@target\grd@@target\relax
        \edef\grd@@start{(\tikztostart)}%
        \tikz@scan@one@point\grd@save@start\grd@@start\relax
        \draw[minor help lines,magenta] (\tikztostart) grid (\tikztotarget);
        \draw[major help lines] (\tikztostart) grid (\tikztotarget);
        \grd@start
        \pgfmathsetmacro{\grd@xa}{\the\pgf@x/1cm}
        \pgfmathsetmacro{\grd@ya}{\the\pgf@y/1cm}
        \grd@target
        \pgfmathsetmacro{\grd@xb}{\the\pgf@x/1cm}
        \pgfmathsetmacro{\grd@yb}{\the\pgf@y/1cm}
        \pgfmathsetmacro{\grd@xc}{\grd@xa + \pgfkeysvalueof{/tikz/grid with coordinates/major step}}
        \pgfmathsetmacro{\grd@yc}{\grd@ya + \pgfkeysvalueof{/tikz/grid with coordinates/major step}}
        \foreach \x in {\grd@xa,\grd@xc,...,\grd@xb}
        \node[anchor=north] at (\x,\grd@ya) {\pgfmathprintnumber{\x}};
        \foreach \y in {\grd@ya,\grd@yc,...,\grd@yb}
        \node[anchor=east] at (\grd@xa,\y) {\pgfmathprintnumber{\y}};
      }
    }
  },
  minor help lines/.style={
    help lines,
    step=\pgfkeysvalueof{/tikz/grid with coordinates/minor step}
  },
  major help lines/.style={
    help lines,
    line width=\pgfkeysvalueof{/tikz/grid with coordinates/major line width},
    step=\pgfkeysvalueof{/tikz/grid with coordinates/major step}
  },
  grid with coordinates/.cd,
  minor step/.initial=.2,
  major step/.initial=1,
  major line width/.initial=2pt,
}
\makeatother

\usepackage{thmtools}
\usepackage{thm-restate}
\usepackage{etoolbox}
\makeatletter
\def\problem@s{}
\newcounter{problems@cnt}

\newcommand{\allproblems}{\problem@s}
\makeatother

\usepackage{times}
\usepackage{amsmath,amsfonts,amssymb,graphicx,mathtools,bm,tcolorbox,relsize}
\usepackage[utf8]{inputenc}
\usepackage[T1]{fontenc}
\usepackage[qm]{qcircuit}
\pgfplotsset{compat=1.18}
\usepackage{braket}
\usepackage{natbib}
\usepackage{multirow}
\usepackage{booktabs}
\usepackage{tabularx}
\usepackage{array}
\usepackage{soul}

\definecolor{colortwo}{rgb}{0.4,0.77,0.17}
\definecolor{colorthree}{rgb}{0.01,0.51,0.93}
\newtcolorbox{tbox}[3][]{%
colframe=#2,colback=#2!10,coltitle=#2!20!black,title={#3},#1}
\nc{\XZX}{\textit{XZX}}
\nc{\YZY}{\textit{YZY}}
\nc{\ZXZ}{\textit{ZXZ}}
\nc{\WZW}{\textit{WZW}}
\nc{\UZU}{\textit{UZU}}
\nc{\QPS}{\mathrm{QPS}}
\nc{\PIS}{\mathrm{PIS}}
\nc{\UQSP}{W_{\omega,\bm\theta,\bm\phi}}
\nc{\UQPP}{V_{\omega,\bm\theta,\bm\phi}}

\nc{\UQSPs}{W}
\nc{\UQPPs}{V}

\nc{\tp}{trigonometric polynomial}
\nc{\tps}{trigonometric polynomials}

\DeclareMathOperator{\sgn}{sgn}
\DeclareMathOperator{\sqw}{sqw}
\DeclareMathOperator{\poly}{poly}
\DeclareMathOperator{\opspan}{span}

\newcommand{\highlight}[1]{ #1 }

\allowdisplaybreaks

\begin{document}
\title{Quantum Phase Processing and its Applications in Estimating Phase and Entropies}
\author{Youle Wang}
\affiliation{Institute for Quantum Computing, Baidu Research, Beijing 100193, China}
\affiliation{School of Software, Nanjing University of Information Science and Technology, Nanjing, 210044, China}
\author{Lei Zhang}
\affiliation{Institute for Quantum Computing, Baidu Research, Beijing 100193, China}
\affiliation{Thrust of Artificial Intelligence, Information Hub, Hong Kong University of Science and Technology (Guangzhou), Guangdong, China}
\author{Zhan Yu}
\affiliation{Institute for Quantum Computing, Baidu Research, Beijing 100193, China}
\affiliation{Centre for Quantum Technologies, National University of Singapore, 117543, Singapore}
\author{Xin Wang}
\email{felixxinwang@hkust-gz.edu.cn}
\affiliation{Institute for Quantum Computing, Baidu Research, Beijing 100193, China}
\affiliation{Thrust of Artificial Intelligence, Information Hub, Hong Kong University of Science and Technology (Guangzhou), Guangdong, China}

\begin{abstract}
Quantum computing can provide speedups in solving many problems as the evolution of a quantum system is described by a unitary operator in an exponentially large Hilbert space. Such unitary operators change the phase of their eigenstates and make quantum algorithms fundamentally different from their classical counterparts. Based on this unique principle of quantum computing, we develop a new algorithmic toolbox ``quantum phase processing'' that can directly apply arbitrary trigonometric transformations to eigenphases of a unitary operator. The quantum phase processing circuit is constructed simply, consisting of single-qubit rotations and controlled-unitaries, typically using only one ancilla qubit. Besides the capability of phase transformation, quantum phase processing in particular can extract the eigen-information of quantum systems by simply measuring the ancilla qubit, making it naturally compatible with indirect measurement. Quantum phase processing complements another powerful framework known as quantum singular value transformation and leads to more intuitive and efficient quantum algorithms for solving problems that are particularly phase-related. As a notable application, we propose a new quantum phase estimation algorithm without quantum Fourier transform, which requires the fewest ancilla qubits and matches the best performance so far. We further exploit the power of our method by investigating a plethora of applications in Hamiltonian simulation, entanglement spectroscopy and quantum entropies estimation, demonstrating improvements or optimality for almost all cases. 
\end{abstract}

\date{\today}
\maketitle

\section{Introduction}

Quantum computer provides a computational framework that can solve certain problems dramatically faster than classical machines. Quantum computing has been applied in many important tasks, including breaking encryption~\cite{shor1997polynomialtime}, searching databases~\cite{grover1996fast}, and simulating quantum evolution~\cite{lloyd1996universal}.
Recent advances in quantum computing show that \emph{quantum singular value transformation} (QSVT) introduced by~\citet{gilyen2019quantum} has led to a unified framework of the most known quantum algorithms~\cite{martyn2021grand}, including amplitude amplification~\cite{gilyen2019quantum}, quantum walks~\cite{gilyen2019quantum}, phase estimation~\cite{martyn2021grand, rall2021faster}, and Hamiltonian simulations~\cite{low2019hamiltonian, lloyd2021hamiltonian,martyn2022efficient,childs2018toward}. This framework can further be used to develop new quantum algorithms such as quantum entropies estimation~\cite{gilyen2020distributional, subramanian2021quantum, gur2021sublinear}, fidelity estimation~\cite{gilyen2022improved}, ground state preparation and ground energy estimation~\cite{lin2020nearoptimal,lin2022heisenberg, dong2022ground}.

The framework of QSVT was originated from a technique called \emph{quantum signal processing} (QSP)~\cite{low2016methodology, low2017optimal}. By interleaving single-qubit signal unitaries and signal processing unitaries, QSP is able to implement a transformation of the signal in $\mathrm{SU}(2)$. There are several conventions of QSP varied by choosing different signal unitaries. In the construction of QSVT, \citet{gilyen2019quantum} chose the signal unitary to be a reflection, then extended the signal unitary to a multi-qubit block encoding with the idea of qubitization~\cite{low2019hamiltonian}, which naturally leads to a polynomial transformation on the singular values of a block-encoded linear operator. The achievable polynomial transformations in QSVT are decided by reflection-based QSP, which has parity constraints or limitations, i.e., it can implement either an even polynomial or an odd one. Thus, to achieve a general transformation in QSVT, one might have to apply techniques such as linear-combination-of-unitaries~\cite{childs2012hamiltonian} and amplitude amplification~\cite{brassard2002quantum, berry2014exponential}, which take extra resources like ancilla qubits. Based on a convention of QSP using $z$-rotation as the signal unitary \cite{haah2019product, chao2020findinga}, \citet{yu2022power} developed an improved version of QSP that overcomes the parity limitation by adding an extra signal processing unitary, which could implement arbitrary complex trigonometric polynomials on one-qubit quantum systems and also shows insights in understanding quantum neural networks.

For the signal unitary being a $z$-rotation, the corresponding trigonometric QSP naturally possesses the ability of processing phase, which indeed plays a central role in many quantum algorithms.  For example, the trick of \emph{phase kickback}, where the phase of the target qubits is kicked back to the ancilla qubit, is intensively used almost everywhere in quantum computing. With the aid of controlled-unitary gates, many quantum algorithms utilize phase kickback to extract information of large unitary operations from phases of ancilla qubits, including the quantum phase estimation~\cite{kitaev1995quantum, cleve1998quantum}, the swap test~\cite{barenco1997stabilization, buhrman2001quantum}, the Hadamard test~\cite{aharonov2009polynomial}, and the one-clean-qubit model~\cite{knill1998power}. 
Hence, it is of great interest and necessity to explore a generalized formalism that could interpret those phase-related quantum algorithms, which may further leads to new quantum algorithms and helps us better exploit the power of quantum signal processing. Consequently, it is natural to investigate and develop the multi-qubit extension of the improved trigonometric QSP in Ref.~\cite{yu2022power}.

In this work, we generalize the trigonometric QSP and propose a novel algorithmic toolbox called \emph{quantum phase processing} (QPP). This toolbox has the ability to apply arbitrary trigonometric transformations to eigenphases of a unitary operator. Besides achieving the eigenphases transformation, QPP is also natively compatible with indirect measurements, enabling it to extract the eigen-information of quantum systems by measuring a single ancilla qubit. We further employ this toolbox to design efficient quantum algorithms for solving various problems. First, we use the idea of binary search to develop an efficient phase estimation algorithm without using quantum Fourier transform, requiring only one ancilla qubit. Such an algorithm can be applied to solve factoring problems and amplitude estimations. 
Second, we show that QPP can be applied to simulate the time evolution under a Hamiltonian $H$ with access to a block encoding of $H$. 
This method is in the same spirit as QSP-based Hamiltonian simulation~\cite{low2017optimal, low2019hamiltonian}, which also matches the optimal query complexity. 
Third, we propose a generic method to estimate quantum entropies, including the von Neumann entropy, the quantum relative entropy and the family of quantum R\'enyi entropies~\cite{Petz1986a}. Despite the fact that QPP could be combined with amplitude estimation to achieve a quadratic speedup, we present algorithms that repeatedly measure the single ancilla qubit to estimate entropies rather than using amplitude estimation, demonstrating its compatibility with indirect measurements. Overall, QPP provides a powerful algorithmic toolbox to exploit quantum applications and delivers a new perspective on understanding and designing quantum algorithms.

The structure of this paper is presented as follows. Section~\ref{sec:framework of QPP} introduces the structure and principal capability of quantum phase processing. In Section~\ref{sec:unitary}, we propose the novel quantum phase search algorithm, then we analyze the performance of the algorithm and make a brief comparison with previous works. Section~\ref{sec:hamiltonian} interprets the method of Hamiltonian simulation in the QPP structure. In Section~\ref{sec:state}, we develop a generic approach for quantum entropies estimation and further showcase the methods of estimating von Neumann entropies, quantum relative entropies and quantum R\'enyi entropies, then we compare our algorithms with prior methods.  Proofs and further discussions of this work are left in the appendix.

\section{Quantum Phase Processing}\label{sec:framework of QPP}

\subsection{Quantum signal processing}
We first review the concept of quantum signal processing (QSP). QSP was introduced by~\citet{low2016methodology}, who showed how to transform a $2\times2$ signal unitary $R_x(x) = e^{ix\sigma_x}$ into a target unitary whose entries are some transformations of the signal $x$. The approach is to apply the signal unitary $R_x(x)$ interleaved with some signal processing unitaries $R_z(\phi)$, i.e.\,
\begin{equation}
    R_z(\phi_0) R_x(x) R_z(\phi_1) R_x(x) \cdots R_x(x) R_z(\phi_k).
\end{equation}
\citet{gilyen2019quantum} modified the signal unitary as a reflection and explicitly showed that the transformation corresponds to a Chebyshev polynomial of the signal $x$. Another common convention of QSP is to choose the signal unitary to be a $z$-rotation $R_z(x)$ with signal processing unitaries being $x$-rotations $R_x(\phi)$~\cite{haah2019product, chao2020findinga}, which corresponds to a trigonometric polynomial of the signal $x$.
Different types of QSP and their relationships are summarized by~\citet{martyn2021grand}.
Observe that both of these two conventions of QSP have constraints on the achievable polynomials: for the Chebyshev QSP, each entry is a polynomial with either even or odd parity; for the trigonometric QSP, each entry is a trigonometric polynomial (in the exponential form) with either real or imaginary coefficients. As a result, the technique of linear-combination-of-unitaries~\cite{childs2012hamiltonian} might be required for these conventions of QSP to implement a general polynomial transformation, which requires extra ancilla qubits. In a recent work, \citet{yu2022power} overcame the constraints by adding an extra signal processing unitary in each iteration so that one could implement arbitrary complex trigonometric polynomial transformation in a single QSP. Our work is heavily based on this improved trigonometric QSP, which is defined as
\begin{equation}\label{eqn:QSP yzzyz decompose}
    \UQSP^L(x) \coloneqq R_z(\omega)R_y(\theta_0)R_z(\phi_0) \prod_{l = 1}^L R_z(x) R_y(\theta_l)R_z(\phi_l),
\end{equation}
where $L\in \NN$ is the number of layers, $\omega \in \RR$, $\bm\theta = (\theta_0, \theta_1, \ldots, \theta_L)\in \RR^{L+1}$ and $\bm\phi = (\phi_0, \phi_1, \ldots, \phi_L)\in \RR^{L+1}$ are parameters. The quantum circuits of different QSP conventions and their realizable polynomials are presented in Table~\ref{tab:qsp circuits}. 


The following lemma characterizes the correspondence between trigonometric QSP and complex trigonometric polynomials. The initial version of Lemma~\ref{lem:trig_qsp} first introduced in~\cite{yu2022power} is in the form of quantum neural networks. Here we restate the lemma in the formalism of QSP without changing the results.

\begin{lemma}[Trigonometric quantum signal processing
~\cite{yu2022power}]\label{lem:trig_qsp}
There exist $\omega \in \RR$, $\bm\theta = (\theta_0, \theta_1, \ldots, \theta_L)\in \RR^{L+1}$ and $\bm\phi = (\phi_0, \phi_1, \ldots, \phi_L)\in \RR^{L+1}$ such that
\begin{equation}\label{eqn:trig form of QSP_}
    \UQSP^L(x) = \begin{bmatrix}
        P(x) & -Q(x)\\
        Q^*(x) & P^*(x)
    \end{bmatrix}
\end{equation}
if and only if Laurent polynomials $P, Q \in \CC[e^{ix/2}, e^{-ix/2}]$ satisfy
\begin{enumerate}
    \item $\deg(P)\leq L$ and $\deg(Q) \leq L$,\label{item:cond1}
    \item $P$ and $Q$ have parity~\footnote{{For a Laurent polynomial $P\in\CC[z,z^{-1}]$, $P$ has parity $0$ if all coefficients corresponding to odd powers of $z$ are $0$, and similarly $P$ has parity $1$ if all coefficients corresponding to even powers of $z$ are $0$.}} $L \bmod 2$,\label{item:cond2}
    \item $\forall x\in \RR$, $\abs{P(x)}^2 + \abs{Q(x)}^2 = 1$.\label{item:cond3}
\end{enumerate}
\end{lemma}

Lemma~\ref{lem:trig_qsp} demonstrates a decomposition of QSP $\UQSP^L(x)$ into complex Laurent polynomials, as well as a construction of QSP from complex Laurent polynomials. {From the second condition of Lemma~\ref{lem:trig_qsp}, it seems that $P$ and $Q$ still have parity constraints, i.e.\ either have parity $0$ or $1$, but in fact Laurent polynomials in $\CC[e^{ix/2}, e^{-ix/2}]$ with parity $0$ are essentially complex trigonometric polynomials in $\CC[e^{ix}, e^{-ix}]$ without parity constraints.} The proof of this theorem also provides an algorithm that calculates angles $\omega, \bm \theta$ and $\bm \phi$ in $\mathcal{O}(\poly(L))$ operations, one can refer to Algorithm~\ref{alg:angle finding} in Appendix~\ref{appendix:angle finding} for more details. There are also many other methods to compute the angles, see e.g., Refs.~\cite{haah2019product, chao2020findinga, silva2022fourierbased, dong2021efficient}. It can be inferred from Lemma~\ref{lem:trig_qsp} that if $P(x)$ satisfies the parity constraint and $\abs{P(x)} \leq 1$ for all $x \in \RR$, then there exists a corresponding $Q(x)$ satisfying the three conditions. The detailed analysis can be found in Appendix~\ref{appendix:angle finding}.


Following the trigonometric QSP construction and decomposition in Ref.~\cite{yu2022power}, we are interested in how to represent the trigonometric polynomial transformation. One way is to project out $P(x)$ from $\UQSP$, i.e., the $\bra{0} \cdot \ket{0}$:

\begin{corollary}\label{coro:trig_qsp_proj}
For any complex-valued trigonometric polynomial $F(x) = \sum_{j=-L}^L c_j e^{ijx}$ with \highlight{$\abs{F(x)} \leq 1$ for all $x \in \RR$}, there exist $\omega \in \RR$ and $\bm\theta,\bm\phi \in \RR^{2L+1}$ such that for all $x\in \RR$,
\begin{equation}\label{eqn:yzzyz_function_}
    \bra{0} \UQSP^{2L}(x) \ket{0} = F(x)
.\end{equation}
\end{corollary}

Moreover, based on the fact that any non-negative real-valued trigonometric polynomial can be decomposed as a product between a Laurent polynomial in $\CC[e^{ix/2}, e^{-ix/2}]$ and its complex conjugate, the trigonometric polynomial can be represented by the expectation value of measuring a Pauli-$Z$ observable with respect to the state $\UQSP^L \ket{0}$:

\begin{corollary}\label{coro:trig_qsp_Z_measurement}
For any real-valued trigonometric polynomial $F(x) = \sum_{j=-L}^L c_j e^{ijx}$ with \highlight{$\abs{F(x)} \leq 1$ for all $x \in \RR$}, there exist $\omega \in \RR$ and $\bm\theta,\bm\phi \in \RR^{L+1}$ such that for all $x\in \RR$,
\begin{equation}\label{eqn:yzzyz_function_expectance}
    f_W(x) \coloneqq \bra{0} \UQSP^L(x)^\dagger Z \UQSP^L(x)\ket{0} = F(x)
.\end{equation}
\end{corollary}

\highlight
{
The key of this corollary is that, for any real-valued trigonometric polynomial $F(x)$ with degree $L$ and satisfies $\abs{F(x)}\leq 1$, we could always find a complex root $P\in\CC[e^{ix/2}, e^{-ix/2}]$ such that $PP^* =(1 + F(x))/2$, as proved in Appendix~\ref{appendix:angle finding}, then the expectation value of measuring $Z$ observable turns to be $F(x)$. However, such a root decomposition does not hold in the polynomial space $\CC[x]$. For instance, let $F(x) = x$, there does not exists $P \in \CC[x]$ such that $PP^* = (1 + F(x))/2$. Thus the Chebyshev QSP~\cite{low2016methodology, low2017optimal,gilyen2019quantum} could not directly use the Pauli-$Z$ measurement to obtain a Chebyshev polynomial.
}

As a summary of this subsection, the construction and functions of the trigonometric QSP are briefly reviewed. Particularly, the trigonometric QSP is capable of implementing arbitrary complex trigonometric polynomials by projection or measuring the Pauli-$Z$ observable in a single-qubit system. These properties are the fundamental reasons why its generalized version can achieve improvements in aspects of phase estimation and entropy estimation, as will be demonstrated in later sections.

\subsection{Processing phases of high-dimensional unitaries}
Although the model of QSP provides a systematical method to make arbitrary polynomial transformations, it only works on a qubit-like quantum system.
\citet{gilyen2019quantum} proposed a multi-qubit lifted version of the Chebyshev QSP, called \emph{quantum singular value transformation} (QSVT), which could efficiently apply Chebyshev polynomial transformations to the singular values of a linear operator embedded in a larger unitary. In this work, we consider a similar extension of the trigonometric QSP and establish a \emph{quantum phase processing} (QPP) algorithmic toolbox that could apply arbitrary trigonometric transformations to eigenphases of a unitary matrix $U$. The structure of QPP generalizes the trigonometric QSP by replacing the input signal $x$ with the phases of a higher-dimensional unitary matrix. For an even $L\in \NN$, angle parameters $\omega \in \RR$ and $\bm\theta, \bm\phi \in \RR^{L + 1}$, we define the quantum phase processor of the $n$-qubit unitary $U$ as
\begin{widetext}
\begin{equation}
    \UQPP^L(U) \coloneqq R_z^{(0)}(\omega)R_y^{(0)}(\theta_0) R_z^{(0)}(\phi_0) \left[
    \prod_{l=1}^{L / 2} 
    \begin{bmatrix}
        U^\dagger & 0 \\
        0 & I^{\otimes n}
    \end{bmatrix} R_y^{(0)}(\theta_{2l - 1})R_z^{(0)}(\phi_{2l - 1}) \begin{bmatrix}
        I^{\otimes n} & 0 \\
        0 & U
    \end{bmatrix} R_y^{(0)}(\theta_{2l})R_z^{(0)}(\phi_{2l})
    \right] \label{eqn:qnn_1}
,\end{equation}
\end{widetext}
where $R_y^{(0)}$ and $R_z^{(0)}$ are rotation gates applied on the first qubit. For an odd $L \in \NN$, we apply an extra layer
\[\begin{bmatrix}
  U^\dagger & 0 \\
  0 & I^{\otimes n}
\end{bmatrix} R_y^{(0)}(\theta_L)R_z^{(0)}(\phi_L)\]
to $\UQPP^{L-1}(U)$. 
The quantum circuit of $\UQPP^L(U)$ is shown as in Fig.~\ref{fig:circuit of QPP}.

\begin{figure*}[htbp]
\[ 
\Qcircuit @C=0.5em @R=0.5em {
\lstick{\ket{0}} & \gate{R_z(\phi_L)} & \gate{R_y(\theta_L)} & \ctrl{1}  & \gate{R_z(\phi_{L-1})} & \gate{R_y(\theta_{L-1})} & \ctrlo{1} & \gate{R_z(\phi_{L-2})} & \gate{R_y(\theta_{L-2})} & \qw & &&\cdots&&& & \ctrlo{1} & \gate{R_z(\phi_{0})} & \gate{R_y(\theta_{0})} &\gate{R_z(\omega)} &\qw &\qw\\
\lstick{} &/^n\qw & \qw & \gate{U}  & \qw & \qw  & \gate{U^\dagger} & \qw & \qw & \qw & &&\cdots&&& & \gate{U^\dagger} & \qw & \qw & \qw &\qw
}
\]
\caption{General circuit for quantum phase processing $\UQPP^L(U)$, here the number of layers $L$ is an even integer.}
\label{fig:circuit of QPP}
\end{figure*}

One could find that QPP simply replaces the signal unitary $R_z(x)$ in QSP with interleaved unitaries controlled-$U$ and controlled-$U^\dagger$. We note that such a construction of using controlled-unitaries was first purposed by \citet{low2017optimal}, namely the ``signal transduction'', which was frequently used in various works~\cite{ low2019hamiltonian, silva2022fourierbased, dong2022ground}. The intuition lying behind the extension is that controlled-$U$ and its inverse are naturally multi-qubit analogs of $R_z$ gates. To better understand how rotation gates in the first ancilla qubit process the phase of the target unitary $U$, we analyze the eigenspace decomposition of QPP:

\begin{lemma} [Eigenspace Decomposition of QPP]\label{lem:qpp eigenspace decomp}
    Suppose $U$ is an $n$-qubit unitary with spectral decomposition 
    \begin{equation}
        U = \sum_{j=0}^{2^n - 1} e^{i \tau_j} \ketbra{\chi_j}{\chi_j}
    .\end{equation}
    For all $L \in \NN$, $\omega \in \RR$ and $\bm\theta,\bm\phi \in \RR^{L+1}$, we have
    \begin{equation}
        \UQPP^L(U) = \bigoplus_{j=0}^{2^n - 1} (e^{-i \tau_j/2})^{L\bmod{2}} \cdot \UQSP^L(\tau_j)_{\mathbb{B}_j}
    \end{equation}
    where $\mathbb{B}_j \coloneqq \{\ket{0, \chi_j}, \ket{1, \chi_j}\}$.
\end{lemma}

The proof of this lemma is deferred to Appendix~\ref{appendix:qpp eigenspace decomp}. Lemma~\ref{lem:qpp eigenspace decomp} shows that the eigenspace of QPP $\UQPP(U)$ coincides with that of the unitary $U$. Using this property, we generalize the single-qubit trigonometric QSP to the multi-qubit QPP that could perform a trigonometric polynomial transformation on the eigenphase of the unitary $U$. In the similar spirit of previous works~\cite{gilyen2019quantum, haah2019product, chao2020findinga}, we could measure the first ancilla qubit and achieve an evolution of the input state upon post-selection of the measurement result being $\ket{0}$, as shown in the follow theorem.


\begin{theorem}[Quantum phase evolution]\label{thm:qpp evolution}
Given an $n$-qubit unitary $U = \sum_{j=0}^{2^n - 1} e^{i \tau_j} \ketbra{\chi_j}{\chi_j}$, for any trigonometric polynomial $F(x) = \sum_{j=-L}^L c_j e^{ijx}$ with $\abs{F(x)} \leq 1$ for all $x \in \RR$, there exist $\omega \in \RR$ and $\bm\theta,\bm\phi \in \RR^{2L+1}$ such that
\begin{equation}
    \UQPP^{2L}(U) = \begin{bmatrix}
        F(U) & \ldots \\
        \ldots & \ldots
    \end{bmatrix}
,\end{equation}
where $F(U) \coloneqq \sum_{j=0}^{2^n - 1} F(\tau_j) \ketbra{\chi_j}{\chi_j}$.
\end{theorem}

The proof of Theorem~\ref{thm:qpp evolution} trivially follows from combining Lemma~\ref{lem:qpp eigenspace decomp} with Corollary~\ref{coro:trig_qsp_proj}, which we defer to Appendix~\ref{appendix:qpp theorems}. The other way is to evaluate trigonometric polynomial on the eigenphases and represent the result by the expectation value of measuring an observable on the first qubit:

\begin{theorem}[Quantum phase evaluation]\label{thm:qpp evaluation}
Given an $n$-qubit unitary $U = \sum_{j=0}^{2^n - 1} e^{i \tau_j} \ketbra{\chi_j}{\chi_j}$ and an $n$-qubit quantum state $\rho$, for any real-valued trigonometric polynomial $F(x) = \sum_{j=-L}^L c_j e^{ijx}$ with \highlight{$\abs{F(x)} \leq 1$ for all $x \in \RR$}, there exist $\omega \in \RR$ and $\bm\theta,\bm\phi \in \RR^{L+1}$ such that $\widehat{\rho} = \UQPP(U) \left(\ketbra{0}{0} \otimes \rho\right)  \UQPP(U)^\dagger$ satisfies
\begin{equation}
    f_V(U) \coloneqq \tr \left[ \left(Z^{(0)}\otimes I^{\otimes n}\right) \cdot \widehat{\rho} \right] = \sum_{j=0}^{2^n - 1} p_{j} F(\tau_j),
\end{equation}
where $p_j = \bra{\chi_j} \rho \ket{\chi_j}$ and $Z^{(0)}$ is a Pauli-$Z$ observable acting on the first qubit.
\end{theorem}


Theorem~\ref{thm:qpp evaluation} is proved by combining Lemma~\ref{lem:qpp eigenspace decomp} and Corollary~\ref{coro:trig_qsp_Z_measurement}, as shown in Appendix~\ref{appendix:qpp theorems}. 
Theorem~\ref{thm:qpp evaluation} shows that QPP is natively compatible with indirect measurements, which could represent the target trigonometric polynomial by probabilities of measuring the ancilla qubit. \highlight{Such a property does not emerge in QSVT, since the Chebyshev QSP typically implements the target polynomial transformation using projection rather than the Pauli-$Z$ measurement.} Chebyshev's inequality dictates that an estimate of the expectation value within an additive error $\d$ can be obtained by measuring the ancilla qubit for $\cO(1/\d^2)$ times. Alternatively, one could apply amplitude estimation~\cite{brassard2002quantum} to estimate the value by calling the QPP circuit for $\cO(1/\d)$ times, which is quadratically more efficiently than classical sampling but with a larger circuit depth. In particular, since the ancilla qubit in QPP naturally works as a flag qubit, we could directly apply the iterative amplitude estimation~\cite{grinko2021iterative} on QPP without using extra qubits.


Theorem~\ref{thm:qpp evolution} and Theorem~\ref{thm:qpp evaluation} lie in the heart of QPP algorithmic toolbox, together demonstrating that QPP is a versatile and flexible toolbox for phase-related problems. First, QPP could act in a similar manner to QSVT but transforming eigenphases of a unitary rather than singular values of an embedded linear operator. Moreover, using block encoding and qubitization~\cite{low2019hamiltonian} enables QPP to process eigenvalues of an embedded operator as well, as presented in later applications. Second, QPP could extract the eigen-information after desired transformations by simply measuring the single ancilla qubit, generalizing many indirect measurement methods like the Hadamard test, which could not be natively achieved by QSVT to the best of our knowledge.

For the sake of notation simplicity, we omit the parameters $L,\omega,\bm\theta$ and $\bm\phi$, writing QSP as $W(x)$ and QPP as $V(U)$ in the rest of this paper. The capabilities of QPP are summarized in Figure~\ref{fig:flow chart}. Next we will show that the structure of QPP is a powerful tool for designing efficient and intuitive quantum algorithms for solving various problems, including quantum phase estimation, Hamiltonian simulation, and quantum entropies estimation.

\begin{figure*}[htbp]
    \centering
    \includegraphics[width=0.8\textwidth]{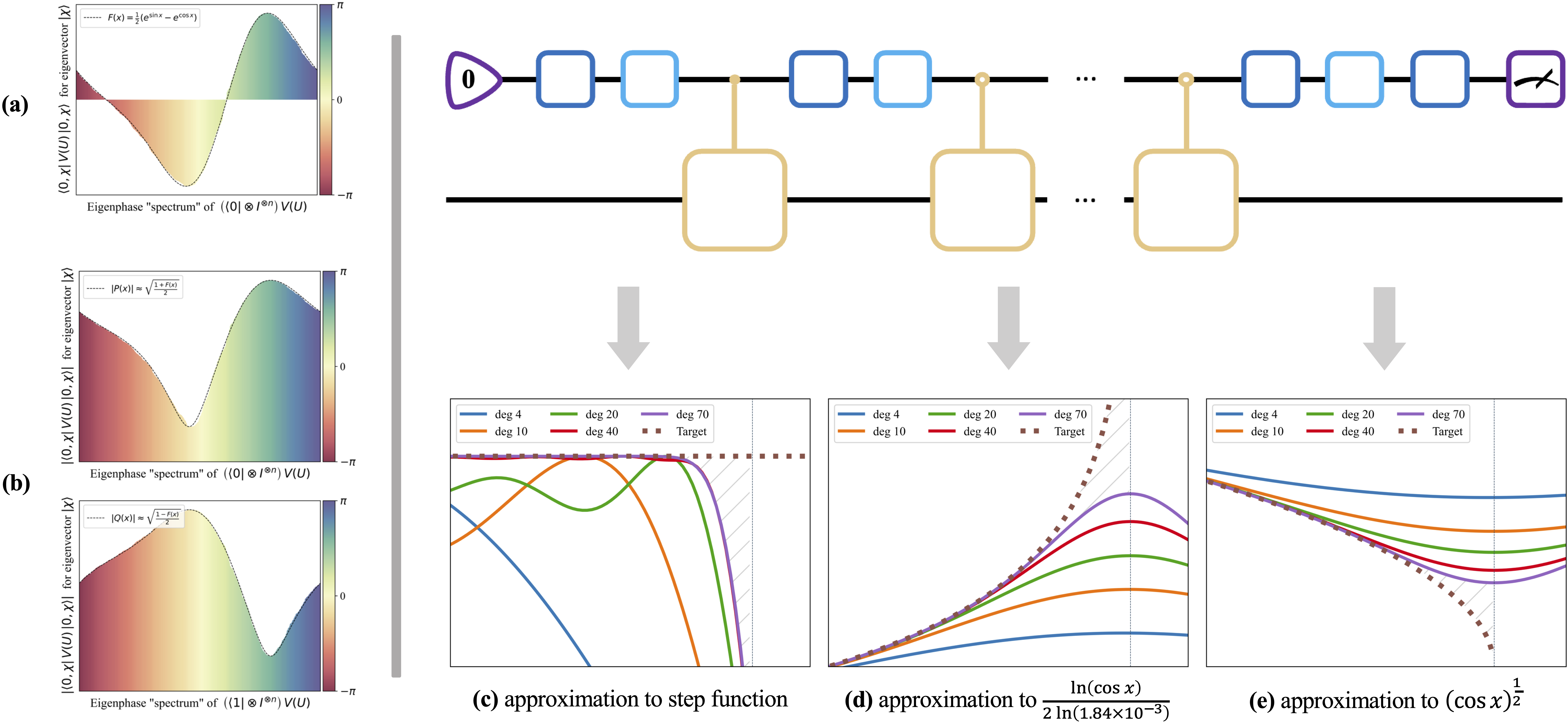}
    \caption{Illustration of the capabilities of quantum phase processing. A trigonometric transformation $F(x) = \frac{1}{2} \left(e^{\sin x} - e^{\cos x} \right)$ of eigenphases of an $n$-qubit unitary $U$ can be represented by a QPP circuit $V(U)$ in two different ways. \textbf{(a)} One approach is to directly simulate the target function. By post-selecting the ancilla qubit to be $\ket{0}$, the function transformations of phases are encoded into the amplitudes of eigenvectors,. \textbf{(b)} The other approach is to retrieve the target function by the expectation value of measuring a Pauli-$Z$ observable on the ancilla qubit. More intuitively, if the input state of circuit $V(U)$ is a maximally mixed state, then the sum of eigenphase transformations is approximately the square difference between two shaded areas in (b) multiplied by a constant. \textbf{(c-e)} Three examples for the QPP circuits to approximate functions near a discontinuity, which will be useful in applications of quantum phase estimation and quantum entropies estimation. In these three subfigures, the quantity ``deg'' refers to the degree of trigonometric polynomials, where every increase of such degree takes an extra layer in a QPP circuit.}
    \label{fig:flow chart}
\end{figure*}


\section{Quantum phase estimation} \label{sec:unitary}

Quantum phase estimation is one of the most important and useful subroutines in quantum computing. The problem of phase estimation is formally defined as follow: Given a unitary $U$ and an eigenstate $\ket{\chi}$ of $U$ with eigenvalue $e^{i\tau}$, estimate the eigenphase $\tau$ up to an additive error $\d$. In this section, we will develop an efficient algorithm for quantum phase estimation based on QPP. Before proceeding, let us do a little warming up to get familiar with the QPP toolbox. We start by considering a simple method of Hadamard test.

\subsection{Warm-up example: the Generalized Hadamard test}\label{sec:Hadamard test}
A common method to estimate the phase of a unitary is the Hadamard test, which solves the following problem: given a unitary $U$ and a state $\ket{\chi}$, estimate $\braket{\chi|U|\chi}$. The Hadamard test uses the measurement result of the ancillary qubit as a random variable whose expected value is the real part $\Re \{\braket{\chi|U|\chi}\}$. It can also estimate the imaginary part $\Im\{\braket{\chi|U|\chi}\}$ by adding a phase gate. We now show that QPP is a generalization of the Hadamard test.

Given a unitary $U$ and a quantum state $\ket{\chi}$ such that $U\ket{\chi} = e^{i \tau} \ket{\chi}$, let $F(x) = \cos(x)$ be the target trigonometric polynomial for QPP. By Theorem~\ref{thm:qpp evaluation}, there exists a single-layer QPP $V(U)$ such that
\begin{equation}
\begin{aligned}
    f_V(U) &= \tr\left[ \bra{0, \chi} V(U)^\dagger \left(Z^{(0)} \otimes I^{\otimes n}\right) V(U) \ket{0, \chi} \right] \\
    &= \cos(\tau) = \Re \{\braket{\chi|U|\chi})\}
,\end{aligned}
\end{equation}
Specifically, by computing the angles $\omega, \bm\theta$ and $\bm\phi$, one can find that the two rotation gates in the first ancilla qubit are essentially Hadamard gates, which means that QPP implements the Hadamard test. Similarly, let the target trigonometric polynomial be $F(x) = \sin(x)$, then QPP can estimate $\sin(\tau) = \Im \{\braket{\chi|U|\chi}\}$. The phase $\tau$ could be obtained from $\cos(\tau)$ and $\sin(\tau)$. Furthermore, one can select trigonometric polynomials other than $\sin(x)$ and $\cos(x)$, which yields a generalization of the Hadamard test. For example, QPP with a trigonometric polynomial $F(x)$ that approximates the function $f(x) = x/2\pi$ could directly estimates the phase $\tau$. 

Although QPP can implement a generalized Hadamard test to estimate the phase, the input state $\ket{\chi}$ is required to be an eigenstate of target unitary. However, quite often we are given a superposition of eigenstates instead of a pure eigenstate, such as in the factoring problem. In addition, measuring the ancilla qubit for $\cO ( 1/\d^2)$ times are necessary to estimate the expected value with an additive error $\d$, despite the fact that the complexity can be improved via amplitude estimation. Can we do better? As demonstrated in the following section, we can further employ QPP to construct a more efficient phase estimation algorithm that accepts a superposition of eigenstates, without using amplitude estimation.

\subsection{Quantum phase searching}
As introduced in Section~\ref{sec:framework of QPP}, QPP can directly process the eigenphases of the target unitary, which allows us to classify the eigenphases. The main idea is to use a trigonometric polynomial to approximate a step function, so that we could utilize QPP to locate the eigenphases by a binary search procedure. We first show that QPP can classify the eigenphases of $U$.

\begin{lemma}[Phase classification]\label{lem:phase classification}
Given a unitary $U = \sum_{j=0}^{2^n - 1} e^{i \tau_j} \ketbra{\chi_j}{\chi_j}$, then for any $\Delta \in (0, \pi)$ and $\eps \in (0, 1)$, there exists a QPP circuit $V(U)$ of $L = \cO \left(\frac{1}{\Delta}\log\frac{1}{\eps}\right)$ layers such that $V (U) \ket{0, \chi_k}$ is
\begin{equation}
    \begin{dcases}
      \sqrt{1 - \eps_k} \ket{0, \chi_k} + \sqrt{\eps_k} \ket{1, \chi_k} & \text{if $\tau_k \in [\Delta, \pi-\Delta)$,}\\
      \sqrt{\eps_k} \ket{0, \chi_k} + \sqrt{1 - \eps_k} \ket{1, \chi_k} & \text{if $\tau_k \in(-\pi + \Delta, -\Delta]$,}
    \end{dcases}
\end{equation}
for $0 \leq k < 2^n$, where $\eps_k \in (0, \eps)$.
\end{lemma}

The proof of Lemma~\ref{lem:phase classification} is deferred to Appendix~\ref{appendix:phase classification}. By phase kickback, measuring the ancilla qubit decides which subinterval the eigenphase $\tau$ belongs to with probability at least $1-\eps$. Next we apply a phase shift $e^{i\z}$ to $U$ to move to the middle point of the designated subinterval, so that $V(e^{i\z}U)$ determines the next subinterval. Using this fascinating property, repeating the binary search procedure shrinks the phase interval until QPP cannot decide next subintervals, i.e.\ $\tau \in [\z_l, \z_r]$ and $\abs{\z_r - \z_l} \approx 2\Delta$. See the phase interval search (PIS) procedure in Algorithm~\ref{alg:interval search} for details.

\begin{figure}[htb]
    \centering
\begin{algorithm}[H]
\caption{Phase Interval Search (PIS)}
\label{alg:interval search}
\begin{algorithmic}[1]
\REQUIRE A unitary $U$, an eigenstate $\ket{\c}$ of $U$ with eigenvalue $e^{i\tau}$, an interval $[\z_l, \z_r]$, a $\Delta \in (0, \frac{1}{2})$, an $\eps \in (0,1)$, and an integer $\cQ$.
\ENSURE An updated interval $[\z_l, \z_r]$ such that $\tau \in [\z_l, \z_r]$ and $\z_r - \z_l = 2(\Delta + \frac{\pi}{2^{\cQ+1}})$.
\FOR{ $j=0\ldots \cQ-1$}
\STATE Set the middle point $\z_m =\frac{\z_l+\z_r}{2}$.
\STATE Construct QPP circuit $\UQPP(e^{-i \z_m} U)$ in Lemma~\ref{lem:phase classification} according to $\Delta$ and $\eps$.
\STATE Apply the circuit to the state $\ket{0, \c}$ and measure the ancilla qubit. If $\z_r - \z_l > 2\pi - 2\Delta$, update
\begin{equation}
    [\z_l, \z_r] \leftarrow
    \begin{dcases}
        [\z_m - \Delta, \z_r + \Delta], & \text{if the outcome is 0} \\
        [\z_l - \Delta, \z_m + \Delta], & \text{if the outcome is 1}
    \end{dcases}
;\end{equation}
otherwise
\begin{equation}
    [\z_l, \z_r] \leftarrow
    \begin{dcases}
        [\z_m - \Delta, \z_r], & \text{if the outcome is 0} \\
        [\z_l, \z_m + \Delta], & \text{if the outcome is 1}
    \end{dcases}
.\end{equation}
\ENDFOR
\STATE Output the interval $[\z_l, \z_r]$.
\end{algorithmic}
\end{algorithm}
\end{figure}

The phase interval search procedure will shrink the phase interval to a length $2(\Delta + \frac{\pi}{2^{\cQ+1}})$ with probability at least $(1-\eps)^\cQ$, where $\cQ$ is the number of repetitions. Now the phase interval is too narrow for phase classification. We apply QPP on $(e^{i\z}U)^d$ for some appropriate integer $d$ so that the binary search procedure can continue to locate the amplified phase $d\tau \in [d\z_l, d\z_r]$. Repeating the entire procedure gives an estimation of phase $\tau$ up to required precision $\d$, as shown in Algorithm~\ref{alg:eigenphase search}.

\begin{figure}[htb]
    \centering
    \begin{algorithm}[H]
\caption{Quantum Phase Search Algorithm (QPS)}
\label{alg:eigenphase search}
\begin{algorithmic}[1]
    \REQUIRE A unitary $U$, an eigenstate $\ket{\c}$ of $U$ with eigenvalue $e^{i\tau}$, a $\Delta \in (0, \frac{1}{2})$, an $\eps \in (0, 1)$ and a $\d \in (0, 1)$.
    \ENSURE A phase $\bar\tau \in \RR$ such that $|\bar{\tau} - \tau| < \delta$.
\STATE Set $\cQ \leftarrow \lceil \log\frac{2\pi}{1 - 2\Delta} \rceil$, $\bar{\Delta} \leftarrow \Delta + \frac{\pi}{2^{\cQ+1}}$, $\cT \leftarrow \lceil \frac{\log\d}{\log\Delta} \rceil$ and $d \leftarrow \lfloor 1/\bar{\Delta} \rfloor$. Initialize the interval $[\z_l, \z_r] \leftarrow [-\pi, \pi]$.
\FOR{$t=0,\ldots, \cT-1$}
\STATE Update the interval by the phase interval search procedure,
\begin{equation}
    [\z_l, \z_r] \leftarrow \PIS \left(U, \ket{\c}, [\z_l, \z_r], \Delta, \frac{\eps}{\cQ\cT}, \cQ\right).
\end{equation}
\STATE Store the middle point of interval $\z_m^{(t)} \leftarrow \frac{\z_l + \z_r}{2}$.
\STATE Update $U \leftarrow \left[e^{-i \z_m^{(t)}} U \right]^d$ and 
\begin{equation}
    [\z_l, \z_r] \leftarrow [d(\z_l - \z_m^{(t)}), d(\z_r - \z_m^{(t)})]
.\end{equation}
\ENDFOR
\STATE Output $\bar\tau \leftarrow \sum_{t=0}^{\cT - 1}\z_m^{(t)} d^{-t}$.
\end{algorithmic}
\end{algorithm}
\end{figure}


Note that Algorithm~\ref{alg:eigenphase search} could accept a superposition of eigenstates as the input, since eigenstates whose eigenvalues disagree with measurement results of the ancilla qubit will be filtered out, and finally the state in the main register converges to a single eigenstate of $U$ at the end of the quantum phase search algorithm. We conclude the above discussions in Theorem~\ref{thm:qps complexity}.

\begin{theorem} [Complexity of Quantum Phase Search] \label{thm:qps complexity}
    Given an $n$-qubit unitary $U$ and an eigenstate $\ket{\c}$ of $U$ with eigenvalue $e^{i\tau}$,  Algorithm~\ref{alg:eigenphase search} can use one ancilla qubit and $\cO \left(\frac{1}{\delta} \log \left(\frac{1}{\eps} \log\frac{1}{\d} \right) \right)$ queries to controlled-$U$ and its inverse to obtain an estimation of $\tau$ up to $\d$ precision with probability at least $1 - \eps$.
\end{theorem}

We can see that the quantum phase search algorithm provides a nearly quadratic speedup compared to the Hadamard test method in Section~\ref{sec:Hadamard test}. More importantly, the phase search algorithm does not require the input state being an eigenstate of $U$, making it more versatile for solving specific problems like the factoring problem. 

\subsection{Comparison to related works}



The quantum phase estimation method originally purposed to solve the Abelian Stabilizer Problem~\cite{kitaev1995quantum, kitaev2002classical} was found to work for general unitaries. The most well-known version of phase estimation method~\cite{kitaev1995quantum} queries the controlled-$U$ for $\cO(\frac{1}{\d})$ times and applies the inverse quantum Fourier transform (QFT)~\cite{coppersmith2002approximate} to estimate the eigenphase of a unitary $U$ with precision $\d$ and success probability at least $4/\pi^2$. The success probability can be boosted to $1-\eps$ by using additional $\cO(\log\frac{1}{\eps})$ ancilla qubits~\cite{cleve1998quantum}. Introducing classical feed-forward process~\cite{kitaev1995quantum, Griffiths1996} can further reduce the number of ancilla to one without increase of circuit depth. Recent studies~\cite{martyn2021grand, rall2021faster} utilize the structure of QSVT to reinterpret phase estimation methods, bringing potential trade-offs among precision, query complexity and number of ancilla qubits. The quantum phase search method proposed in this work is based on the intuitive idea of binary search, which is fundamentally different from the previous QFT-based algorithms. {\citet{dong2022ground} proposed a phase estimation algorithm, in a similar spirit of ours, to the case $U\ket{\psi}=e^{-i\lambda}\ket{\psi}$, where $\lambda\in[\eta, \pi-\eta]$ for some constant $\eta>0$. However, this prior condition on the eigenvalue is usually not satisfied for an arbitrary unitary and initial state.} Here we compare our phase search method to some previous phase estimation algorithms with respect to the query complexity and number of ancilla qubits under the same precision and success probability, which is shown in Table~\ref{table:QPE algs}. One can see that the quantum phase search method achieves the best query complexity while requiring the least ancilla qubits, turning out to be an efficient phase estimation algorithm. 

\begin{table*}[tp]
\centering
\setlength{\tabcolsep}{1em}
\begin{tabular}{lccccc}
\toprule
Methods for QPE & Queries to controlled-$U$ & \# of ancilla qubits & Success probability & Precision \\
\midrule
QFT-based~\cite{kitaev1995quantum, cleve1998quantum, Nielsen2010} & $ \cO \left( \frac{1}{\d}\left(1 + \frac{1}{\eps}\right) \right)$ & $ \ceil*{\log{\frac{1}{\d}} + \log(2+\frac{1}{\eps})}$ & {} & {} \\
\addlinespace
Semi-classical QFT-based~\cite{kitaev1995quantum, Griffiths1996} & $\cO \left( \frac{1}{\d}\left(1 + \frac{1}{\eps}\right) \right)$ & $1$ & \multirow{3}{*}{$ 1-\eps$} & \multirow{3}{*}{$ \d$} \\
\addlinespace
QSVT-based~\cite{martyn2021grand, rall2021faster} & $ \widetilde\cO \left(\frac{1}{\d} \log\left(\frac{1}{\eps}\right) \right)$ & $ \ceil*{\log{\frac{1}{\d}}} + 3$ (or $3$) & {} & {} \\
\addlinespace
QPP-based (in Theorem~\ref{thm:qps complexity}) & \multicolumn{1}{c}{$ \widetilde\cO \left(\frac{1}{\d} \log\left(\frac{1}{\eps}\right) \right)$} & $1$ & {} & {}\\
\bottomrule
\end{tabular}
\caption{Comparison of QPE Algorithms. In the query complexity, the $\widetilde\cO$ notation omits $\log\log$ factors. The number of ancilla qubits is the total number of qubits used other for the system of $U$.}\label{table:QPE algs}
\end{table*}


\section{Quantum entropy estimation} \label{sec:state}
Quantum entropy is used to characterize the randomness and disorder of a quantum system, which has various theoretical and experimental applications of relevance. Estimating the entropy of a quantum system is an important problem in quantum information science. Classical methods of estimating the quantum entropies require the density matrix of a quantum state, which is costly, especially when the size of system is large. Recent works proposed quantum algorithms that could efficiently estimate quantum entropies~\cite{gilyen2020distributional,subramanian2021quantum,gur2021sublinear}, showing potential quantum speedups over the classical methods. The motivating idea behind these quantum approaches is the purified quantum query model~\cite{low2019hamiltonian}, which prepares a purification of a mixed state $\rho$. {The purified query model can apply to cases where the states are generated by a quantum circuit and have applications in many tasks \cite{belovs2019quantum,gilyen2020distributional}.} Formally, the purified quantum query oracle of a mixed state $\rho$ is a unitary $U_\rho$ acting as
\begin{equation}\label{eqn:purified query}
    U_\rho \ket{0}_A\ket{0}_B = \ket{\Psi_\rho}_{AB} = \sum_{j=0}^{2^n-1} \sqrt{p_j} \ket{\psi_j}_A\ket{\phi_j}_B
\end{equation}
such that $\tr_B(\ketbra{\Psi_\rho}{\Psi_\rho}) = \rho$, where $\{\ket{\psi_j}\}$ and $\{\ket{\phi_j}\}$ are sets of orthonormal states on the system $A$ and $B$ respectively. Using the qubitization method in~\cite{low2019hamiltonian}, such an oracle model can be used to build a qubitized block encoding $\widehat{U}_\rho$ of the target state $\rho$. We show the detailed construction of $\widehat{U}_\rho$ in Appendix~\ref{appendix:dm block enc}. 


Consequently, it is reasonable and compelling to investigate whether QPP, the structure designed for unitary phase processing, could contribute to the improvement of quantum algorithms for quantum entropy estimation. Note that quantum entropies of a quantum state $\rho$ can be interpreted as the corresponding classical entropies of the eigenvalues of $\rho$. If one could find \tps{} $F(x)$ that approximate the classical entropic functions,
then quantum entropies can be naturally estimated via phase evaluation of $\widehat{U}_\rho$ in Theorem~\ref{thm:qpp evaluation} by the spectral correspondence between $\rho$ and $\widehat{U}_\rho$. Specifically, the following theorem is the basic principle of the QPP-based quantum entropies estimation, and the proof of which is deferred to Appendix~\ref{appendix:state entropy estimation}.

\begin{theorem} \label{thm:state entropy estimation}
    Let $\ket{\Psi_\rho}_{AB}$ be a purification of an $n$-qubit state $\rho$ and $\widehat{U}_\sigma$ be a qubitized block encoding of an $n$-qubit state $\sigma$ with $m$ ancilla qubits. For any real-valued polynomial $f(x) = \sum_{k=0}^L c_j x^k$ with \highlight{$\abs{F(x)} \leq 1$ for all $x \in \RR$}, there exists a QPP circuit $\UQPPs(\widehat{U}_\sigma)$ of $L$ layers such that 
    \begin{align}
        \langle Z^{(0)} \rangle_{\ket{\Phi}} = \tr \left(\rho f(\sigma) \right)
    ,\end{align}
    where $\ket{\Phi} = \left(\UQPPs(\widehat{U}_\sigma) \otimes I_B\right)\ket{0^{\otimes (m + 1)}} \ket{\Psi_\rho}_{AB}$ and the polynomial on a quantum state is defined as $f(\sigma) = \sum_{k=0}^L c_j \sigma^k$.
\end{theorem}

Theorem~\ref{thm:state entropy estimation} shows that one could measure the value of $\tr(\rho f(\sigma))$ as the $Z$ expectation value of the ancilla qubit. An estimate of the expectation value within an additive error $\d$ can be obtained by measuring the ancilla qubit for $\cO(1/\d^2)$ times. Moreover, we could directly apply the iterative amplitude estimation~\cite{grinko2021iterative} on QPP to achieve a quadratic speedup without using extra qubits. For clarity, we present the QPP circuit of quantum entropy estimation in Figure~\ref{fig:circuit of QPP for state}. We also note that Theorem \ref{thm:state entropy estimation} can be applied to extract many information-theoretic properties of quantum states other than quantum entropies.

\begin{figure*}[htbp]
\[ 
\Qcircuit @C=0.5em @R=1.5em {
\lstick{\ket{0}} & \qw & \gate{R_z} & \gate{R_y} & \ctrl{1}  & \gate{R_z} & \gate{R_y} & \ctrlo{1} & \gate{R_z} & \gate{R_y} & \qw & &&\cdots&&& & \ctrlo{1} & \gate{R_z} & \gate{R_y} &\gate{R_z} &\qw &\meter\\
\lstick{\ket{0^{\otimes m}}} & /\qw & \qw & \qw & \multigate{1}{\widehat{U}_\sigma}  & \qw & \qw  & \multigate{1}{\widehat{U}_\sigma^\dagger} & \qw & \qw & \qw & &&\cdots&&& & \multigate{1}{\widehat{U}_\sigma^\dagger} & \qw & \qw & \qw &\qw \\
\lstick{\ket{0^{\otimes n}}_A} & /\qw & \multigate{1}{U_\rho} & \qw & \ghost{\widehat{U}_\sigma}  & \qw & \qw  & \ghost{\widehat{U}_\sigma^\dagger} & \qw & \qw & \qw & &&\cdots&&& & \ghost{\widehat{U}_\sigma^\dagger} & \qw & \qw & \qw &\qw \\
\lstick{\ket{0}_B} & /\qw & \ghost{U_\rho} & \qw & \qw & \qw & \qw  & \qw & \qw & \qw & \qw & &&\cdots&&& & \qw & \qw & \qw & \qw &\qw
}
\]
\caption{The QPP circuit in Theorem~\ref{thm:state entropy estimation} for quantum entropies estimation.}
\label{fig:circuit of QPP for state}
\end{figure*}

In this section, to demonstrate the power of QPP, we utilize the generic method in Theorem~\ref{thm:state entropy estimation} to estimate the most fundamental entropic functionals for quantum systems, including the von Neumann entropy, the quantum relative entropy, and the family of quantum R\'enyi entropies.

\subsection{von Neumann and quantum relative entropy estimation}

The von Neumann entropy~\cite{von1932mathematische} is a generalization of the Shannon entropy from the classical information theory to quantum information theory. For an $n$-qubit quantum state $\rho$, the von Neumann entropy is defined as follows
\begin{equation}
    S(\rho) = - \tr(\rho \ln \rho).
\end{equation}
Let $\{p_j\}_j$ be the eigenvalues of $\rho$, then the von Neumann entropy is the same as the Shannon entropy of the probability distribution $\{p_j\}_j$,
\begin{equation}
    S(\rho) = - \sum_{j=0}^{2^n - 1} p_j \ln p_j.
\end{equation}
Recall from the qubitization technique that {partial} eigenphases of the qubitized block encoding $\widehat{U}_\rho$ are {given by} $\pm \tau_j = \pm \arccos(p_j)$, then we have $p_j = \cos(\pm \tau_j) \in [0, 1]$. Here we assume the non-zero eigenvalues are lower bounded by some $\gamma > 0$. Then by Theorem~\ref{thm:state entropy estimation}, the main idea of using QPP to estimate the von Neumann entropy is to find a polynomial $f(x)$ that approximates the function $\ln(x)$ with some appropriate scale on the interval $[\gamma, 1]$, and $|f(x)| \leq 1$ for $x \in [-1, 1]$. Particularly, the polynomial $f(x)$ could be obtained from the Taylor series of $\ln(x)$. The overall result is stated in the following theorem.

\begin{theorem}[von Neumann entropy estimation] \label{thm:von Neumann complexity}
    Given a purified quantum query oracle $U_\rho$ of a state $\rho$ whose non-zero eigenvalues are lower bounded by $\gamma > 0$, there exists an algorithm that estimates $S(\rho)$ up to precision $\eps$ with high probability by measuring a single qubit, querying $U_\rho$ and $U_\rho^\dagger$ for $\cO\left(\frac{1}{\gamma \eps^2}\log^2(\frac{1}{\gamma})\log(\frac{\log(1/\gamma)}{\eps})\right)$ times. Moreover, using amplitude estimation improves the query complexity to $\cO\left(\frac{1}{\gamma \eps}\log(\frac{1}{\gamma})\log(\frac{\log(1/\gamma)}{\eps})\right)$.
\end{theorem}

In particular, the dependence on $\gamma$ can be translated to the rank (or dimension) of the density matrix, from which we have the following corollary.
\begin{corollary} \label{coro:von Neumann complexity no gamma}
    Given a purified quantum query oracle $U_\rho$ of a state $\rho$ whose rank is $\kappa \leq 2^n$, there exists an algorithm that estimates $S(\rho)$ up to precision $\eps$ with high probability by measuring a single qubit, querying $U_\rho$ and $U_\rho^\dagger$ for $\cO\left( \frac{\kappa}{\eps^3} \log^3 \left(\frac{\kappa}{\eps}\right)\log\left(\frac{1}{\eps}\right)\right)$ times. Moreover, using amplitude estimation improves the query complexity to $\cO\left( \frac{\kappa}{\eps^2} \log^2 \left(\frac{\kappa}{\eps}\right) \log\left(\frac{1}{\eps}\right)\right)$.
\end{corollary}

We defer the proofs to Appendix~\ref{appendix:guarantee for neumann and relative}. Note that the estimation of the quantum relative entropy between states $\rho$ and $\sigma$, i.e.\ 
\begin{equation}
    \infdiv{\rho}{\sigma} = -\tr(\rho\ln\sigma) - S(\rho),
\end{equation}
immediately follows from the above analysis. In particular, we only need to apply QPP on a qubitized block encoding of $\sigma$ to estimate $\tr(\rho\ln\sigma)$ if the relative entropy is finite. The result of quantum relative entropy estimation is shown in Theorem~\ref{thm:relative complexity}.

\begin{theorem}[Quantum relative entropy estimation]
\label{thm:relative complexity}
 Given purified quantum query oracles $U_\rho$ and $U_\sigma$ of states $\rho$ and $\sigma$, respectively, such that their non-zero eigenvalues are lower bounded by $\gamma > 0$ and the kernel of $\sigma$ has trivial intersection with the support of $\rho$, there exists an algorithm that estimates $\infdiv{\rho}{\sigma}$ up to precision $\eps$ with high probability, querying $U_\rho$, $U_\sigma$ and their inverses for $\cO\left(\frac{1}{\gamma \eps^2}\log^2(\frac{1}{\gamma})\log(\frac{\log(1/\gamma)}{\eps})\right)$ times. Moreover, using amplitude estimation improves the query complexity to $\cO\left(\frac{1}{\gamma \eps}\log(\frac{1}{\gamma})\log(\frac{\log(1/\gamma)}{\eps})\right)$.
\end{theorem}

\subsection{Quantum R\'enyi entropy estimation}
The quantum R\'enyi entropy~\cite{Petz1986a} is a quantum version of the classical R\'enyi entropy that was first introduced in~\cite{renyi1961measures}. For $\alpha\in(0,1)\cup(1,\infty)$, the quantum $\a$-R\'enyi entropy of an $n$-qubit quantum state $\rho$ is defined as follows:
\begin{equation}
    S_{\alpha}(\rho)=\frac{1}{1-\alpha}\log\tr(\rho^\alpha).
\end{equation}
Let $\{p_j\}_j$ be the eigenvalues of $\rho$, then the quantum $\a$-R\'enyi entropy reduces to the $\a$-R\'enyi entropy of the probability distribution $\{p_j\}_j$,
\begin{equation}
    S_{\alpha}(\rho)=\frac{1}{1-\alpha}\log \left(\sum_{j=0}^{2^n - 1} p_j^\alpha \right).
\end{equation}
Similarly, we assume all non-zero eigenvalues are greater than some $\gamma > 0$. The method of R\'enyi entropy estimation, based on Theorem~\ref{thm:state entropy estimation}, is in the same spirit of estimating the von Neumann entropy; the only difference is that we now aim to find a polynomial $f(x)$ that approximates the function $x^{\a - 1}$ for any $\alpha>0$ and $\alpha\neq1$ on the interval $[\gamma, 1]$. The exponent is $\a - 1$ because we can write $\tr(\rho^\a)$ as $\tr(\rho \cdot \rho^{\a-1})$, and the isolated $\rho$ comes from the input state of QPP. 

When $\alpha>0$ is a non-integer, the polynomial could be given by separately considering the integer part and decimal part of $\alpha - 1$. 
Thus we only need to find a polynomial that approximates $x^{\a-1}$ for $\alpha \in (0, 1)$ on the interval $x \in [\gamma, 1]$, which could be obtained from the Taylor series of $x^{\a-1}$. We present the results in the following theorem and more discussions in Appendix~\ref{appendix:guarantee for Renyi}.

\begin{theorem} [Quantum R\'enyi entropy estimation for real $\alpha$]
\label{thm:real Renyi complexity}
 Given a purified quantum query oracle $U_\rho$ of a state $\rho$ whose non-zero eigenvalues are lower bounded by $\gamma > 0$, there exists an algorithm that estimates $S_\alpha(\rho)$ up to precision $\eps$ with high probability by measuring a single qubit, querying $U_\rho$ and $U_\rho^\dagger$ for a total number of times of
\begin{equation}
    \begin{dcases}
        \widetilde\cO \left(\frac{1}{\gamma^{3 - 2\a} \eps^2} \cdot \eta^2 \right), & \textrm{if } \a \in (0, 1); \\[10pt]
        \widetilde\cO \left(\frac{\a\gamma + 1}{\gamma \eps^2} \cdot \eta^2 \right), & \textrm{if } \a \in (1, \infty);
    \end{dcases}
\end{equation}
where $\eta = \frac{\tr(\rho^\a)^{-1}}{\abs{1 - \a}}$.
Moreover, using quantum amplitude estimation improves the query complexity to
\begin{equation}
    \begin{dcases}
        \widetilde\cO \left(\frac{1}{\gamma^{2 - \a} \eps} \cdot \eta \right), & \textrm{if } \a \in (0, 1). \\[10pt]
        \widetilde\cO \left(\frac{\a\gamma + 1}{\gamma \eps} \cdot \eta \right), & \textrm{if } \a \in (1, \infty).
    \end{dcases}
\end{equation}
Here the $\widetilde\cO$ notation omits logarithmic factors.
\end{theorem}

Similarly, we provide a method to estimate $S_\a(\rho)$ without information of $\gamma$ in Appendix~\ref{appendix:guarantee for Renyi}.
When $\alpha$ is an integer, the function $x^{\a - 1}$ naturally turns to be a normalized polynomial so that approximation error does not exist by Theorem~\ref{thm:state entropy estimation}. {In this case, the dependence on the threshold $\gamma$ can be further improved.}

\begin{theorem}[Quantum R\'enyi entropy estimation for integer $\alpha$]
\label{thm:integer Renyi complexity}
Suppose $\alpha>1$ is a positive integer, there exists an algorithm that estimates $S_\a(\rho)$ up to precision $\eps$ with high probability by measuring a single qubit, querying $U_\rho$ and $U_\rho^\dagger$ for $\cO\left(\frac{\alpha\tr(\rho^\a)^{-2}}{ \eps^2}\right)$ times.
Moreover, using amplitude estimation improves the query complexity to $\cO\left(\frac{\alpha\tr(\rho^\a)^{-1}}{\eps}\right)$.
\end{theorem}

Note that this method of computing $S_\a(\rho)$ for $\alpha \in \NN_+$ naturally establishes an efficient algorithm for entanglement spectroscopy, a task of obtaining the entanglement of a quantum state. Consider a bipartite pure state $\ket{\Psi_\rho}_{AB}$ in Eq.~\eqref{eqn:purified query}, the entanglement between systems $A$ and $B$ can be characterized by the eigenvalues of the reduced density operator $\rho = \tr_B(\ketbra{\Psi_\rho}{\Psi_\rho})$. Specifically, one needs to compute the $\tr(\rho^k)$ for $k= 1, \ldots, k_{\max}$ to estimate $k_{\max}$ largest eigenvalues of $\rho$ by the Newton-Girard method~\cite{johri2017entanglementa, subasi2019entanglement, yirka2021qubitefficient}.

\subsection{Comparison to related works}
As introduced above, we utilize the structure of QPP to estimate quantum entropies based on the purified quantum query model. Here we briefly mention some closely related works on quantum entropy estimation under a similar setting. 
For von Neumann entropy $S(\rho)$, \citet{gilyen2020distributional} proposed an efficient quantum algorithm based on QSVT and amplitude estimation that achieves a near-linear query complexity and an additive error $\eps$.  Another work by~\citet{gur2021sublinear} utilized the quantum singular value estimation \cite{kerenidis2016quantum} and amplitude estimation to implement an algorithm with a sublinear query complexity up to a multiplicative error bound.
By contrast, our algorithm could estimate the result by measuring the first ancilla qubit, which has a slightly worse query complexity in the worst case but a smaller circuit size. Nevertheless, by Corollary~\ref{coro:von Neumann complexity no gamma}, our complexity can depend on the rank of the density matrix, which will be further improved for low-rank cases. Note that one could apply amplitude estimation on QPP without using extra qubits, which might be required for previous QSVT-based algorithms. As a result, the QPP-based algorithms allow us to flexibly consider the trade-off between the query complexity and the circuit depth in practical applications. 

\begin{table*}[tp]
\centering
\setlength{\tabcolsep}{1em}
\begin{tabular}{lccc}
\toprule
Methods for $S(\rho)$ estimation & Total queries to $U_\rho$ and $U_\rho^\dagger$ & Queries per use of circuit\\
\midrule
QSVT-based with QAE (\cite{gilyen2020distributional}) & $\widetilde\cO(\frac{d}{\eps^{1.5}})$ & $\widetilde\cO(\frac{d}{\eps^{1.5}})$ \\
\addlinespace
QPP-based (assumes rank, in Corollary~\ref{coro:von Neumann complexity no gamma}) & $\widetilde\cO(\frac{\kappa}{\eps^3})$ & $\widetilde\cO(\frac{\kappa}{\eps})$  \\
\addlinespace
QPP-based with QAE (assumes rank, in Corollary~\ref{coro:von Neumann complexity no gamma})  & $\widetilde\cO(\frac{\kappa}{\eps^2})$ &  $\widetilde\cO(\frac{\kappa}{\eps^2})$  \\
\addlinespace
QPP-based (in Theorem~\ref{thm:von Neumann complexity})  & $\widetilde\cO(\frac{1}{\gamma \eps^2})$ & $\widetilde\cO(\frac{1}{\gamma})$  \\
\addlinespace
QPP-based with QAE (in Theorem~\ref{thm:von Neumann complexity}) & $\widetilde\cO(\frac{1}{\gamma \eps})$ &  $\widetilde\cO(\frac{1}{\gamma\eps})$  \\
\bottomrule
\end{tabular}
\caption{Comparison of algorithms on estimating von Neumann entropy within additive error. Here the $\widetilde\cO$ notation omits $\log$ factors, $\gamma > 0$ is the lower bound of eigenvalues, $\kappa > 0$ is the rank of the state $\rho \in \CC^{d \times d}$, and $\eps$ is the additive error of estimating $S(\rho)$. QAE is short for quantum amplitude estimation.} 
\label{table:von Neumann algs}
\end{table*}

With regard to the family of quantum $\a$-R\'enyi entropies $S_\a(\rho)$, when $\a$ is an integer, QPP establishes an efficient algorithm for entanglement spectroscopy. Compared to previous algorithms for entanglement spectroscopy~\cite{johri2017entanglementa, subasi2019entanglement}, the QPP-based algorithm significantly reduces the circuit width from $\Theta(n\a)$ to $4n+1$ without using qubit resets as in~\cite{yirka2021qubitefficient}. For a more general case that $\a$ is not an integer, \citet{subramanian2021quantum} introduced a quantum algorithm that combines the QSVT technique and the $\mathsf{DQC1}$ (Deterministic Quantum Computation with one clean qubit) method. Their algorithm estimates $S_\alpha(\rho)$ with an additive error by measuring a single qubit, using an expected total number $\widetilde\cO(\sfrac{d^2}{\gamma\eps^2})$ queries to the purified quantum oracle, where $d = 2^n$ is the dimension of $\rho$. Our QPP-based approaches improve the results in~\cite{subramanian2021quantum} in terms of the dependence on the dimension. For instance, for $\alpha>1$, our algorithms based on single-qubit measurement require a query complexity of $\widetilde\cO(\sfrac{\alpha}{\gamma\eps^2 })$. {The main reason there is a speedup factor of $d^2$ is that the $\mathsf{DQC1}$ method requires a maximally mixed state $\frac{1}{d}I$ as the input state, whereas the QPP-based method uses $\rho$ as the input state.} Moreover, as we mention before, QSVT is not natively compatible with indirect measurement, one needs to utilize an extra ancilla qubit to control the QSVT circuit in order to implement indirect measurement, as shown in~\cite{subramanian2021quantum}.
Similar to the earlier description, we could leverage quantum amplitude estimation to achieve a better query complexity. More detailed comparisons are shown in Table~\ref{table:Renyi algs}.

\begin{table*}[tp]
\centering
\setlength{\tabcolsep}{1em}
\begin{tabular}{lcccc}
\toprule
\multirow{2}{*}[-2pt]{Methods for $S_\a(\rho)$ estimation} & \multicolumn{3}{c}{Total queries to $U_\rho$ and $U_\rho^\dagger$}\\
\cmidrule{2-4}
& $\a \in (0, 1)$ & $\a \in (1, \infty)$ &  $\a \in \NN_+$\\
\midrule
QSVT-based with $\mathsf{DQC1}$ (\cite{subramanian2021quantum}) & $\widetilde\cO \left(\frac{d^2}{\gamma \eps^2} \cdot \eta^2\right)$ & $\widetilde\cO \left(\frac{d^2}{\gamma \eps^2} \cdot \eta^2\right)$ & $\cO \left(\frac{d^2}{\eps^2} \cdot \a^3 \eta^2\right)$  \\
\addlinespace
QPP-based (in Theorem~\ref{thm:real Renyi complexity} and~\ref{thm:integer Renyi complexity}) & $\widetilde\cO \left(\frac{1}{\gamma^{3 - 2\a} \eps^2} \cdot \eta^2\right)$ & $\widetilde\cO \left(\frac{\a\gamma + 1}{\gamma \eps^2} \cdot \eta^2\right)$ & $\cO \left(\frac{1}{\eps^2} \cdot \a \eta^2\right)$ \\
\addlinespace
QPP-based with QAE (in Theorem~\ref{thm:real Renyi complexity} and~\ref{thm:integer Renyi complexity}) & $\widetilde\cO \left(\frac{1}{\gamma^{2 - \a} \eps} \cdot \eta\right)$ & $\widetilde\cO \left(\frac{\a\gamma + 1}{\gamma \eps} \cdot \eta\right)$ & $\cO \left(\frac{1}{\eps} \cdot \a \eta\right)$ \\
\bottomrule
\end{tabular}
\caption{Comparison of algorithms on estimating quantum $\a$-R\'enyi entropies within additive error for different $\a$. Here the $\widetilde\cO$ notation omits $\log$ factors, $\gamma > 0$ is the lower bound of eigenvalues of a mixed state $\rho \in \CC^{d \times d}$, $\eta \coloneqq \frac{\tr(\rho^\a)^{-1}}{|1-\a|}$ is the quantity depending on $\a$ and $\rho$, and $\eps$ is the additive error of estimating $S_\a(\rho)$. QAE is short for quantum amplitude estimation.}\label{table:Renyi algs}
\end{table*}

The polynomial transformation implemented by QSVT lies in amplitudes of the outcome state, which could not be obtained by indirect measurements of the ancilla qubit in QSVT, thus most QSVT-based entropies estimation algorithms estimate the value either by applying amplitude estimation or combining with the $\mathsf{DQC1}$ model, and both of these methods increase the circuit size. Another approach is using a polynomial to estimate the square root of the function $\sqrt{f(x)}$, as shown by~\citet{wang2022new}. This approach makes QSVT compatible with indirect measurements, since the approximated function now can be represented by the probability of measuring, which is similar as in QPP. However, the problem is that sometimes $\sqrt{f(x)}$ could be more difficult to approximate than $f(x)$. For example, $f(x) = x$ is just a simple one-term polynomial, whereas $\sqrt{f(x)} = \sqrt{x}$ takes much more terms to precisely approximate. Thus such a method presented in~\cite{wang2022new} may lead to even worse complexity than previous ones in~\cite{gilyen2020distributional, subramanian2021quantum, gur2021sublinear}. We defer further discussion and comparison to Appendix~\ref{appendix:entropy table}.


\section{Hamiltonian simulation} \label{sec:hamiltonian}
{This section aims to investigate the application of our phase processing circuits in simulating the dynamics of a quantum system. Unsurprisingly, our algorithms can reproduce the optimal results introduced in \cite{low2017optimal}.} {Specifically, for Hamiltonian $H$}, given evolution time $t$ and simulation error $\eps$, the goal is to simulate the time evolution $e^{-iHt}$ with error $\eps$, i.e.,\ produce a unitary $U$ such that $\norm{U - e^{-iHt}}\leq \eps$. {Usually}, this is accomplished by designing a quantum circuit to simulate the operator $e^{-iHt}$ with high precision. {In the setup, we assume a block encoding of the Hamiltonian, which is a commonly used input model appearing in \cite{low2019hamiltonian, chakraborty2019power,gilyen2019quantum}. Before presenting the results, we recall the definition of block encoding.}


\subsection{Block Encoding} \label{sec:block enc intro}

A block encoding of a matrix $A\in \CC^{2^n \times 2^n}$ with spectral norm $\norm{A} \leq 1$ is a unitary matrix $U_A$ such that the upper-left block of the matrix is $A$, i.e.\ 
\begin{equation}
    U_A=\begin{bNiceMatrix}[columns-width=1.5em]
        A  & \cdot\\
        \cdot & \cdot
    \end{bNiceMatrix}.
\end{equation}
Then we can write $A = \left(\bra{0^{\otimes m}} \otimes I^{\otimes n}\right) U_A \left(\ket{0^{\otimes m}} \otimes I^{\otimes n}\right)$, where $m$ denotes the number of ancilla qubits in the block encoding. In other words, for any state $\ket{\psi} \in \CC^N$, we have $\left(\bra{0^{\otimes m}} \otimes I^{\otimes n}\right) U_A \ket{0^{\otimes m}, \psi} = A\ket{\psi}$. In particular, let $\ket{\psi_\lambda}$ be an eigenvector of $A$ with an eigenvalue $\lambda$, then we will have a state 
\begin{align}
    U_A \ket{0^{\otimes m},\psi_\lambda}=\lambda \ket{0^{\otimes m}, \psi_\lambda}+\sqrt{1-\lambda^2}\ket{\bot_\lambda},
    \label{eqn:psi}
\end{align}
where $\ket{\bot_\lambda}$ denotes a state orthogonal to $\ket{0^{\otimes m}, \psi_\lambda}$. In the above equation, we absorb the relative phase into states and ignore the global phase. To associate the spectrum of $A$ and its block encoding $U$, it has to be ensured that the subspace $\opspan\{\ket{0^{\otimes m}, \psi_\lambda}, U \ket{0^{\otimes m}, \psi_\lambda}\}$ is invariant under $U_A$. However, this is generally not true for an arbitrary block encoding. To resolve this issue, \citet{low2019hamiltonian} proposed a so-called ``qubitization'' technique that uses controlled-$U_A$ and controlled-$U^\dagger_A$ once to implement a qubitized block encoding $\widehat{U}_A$ preserving the subspace $\opspan\{\ket{0^{\otimes (m+1)}, \psi_\lambda}, \widehat{U}_A \ket{0^{\otimes (m+1)}, \psi_\lambda}\}$. Write
\begin{equation} \label{eqn:qubitization for block encoding}
    \widehat{U}_A \ket{0^{\otimes (m + 1)},\psi_\lambda}=\lambda \ket{0^{\otimes (m + 1)}, \psi_\lambda}+\sqrt{1-\lambda^2}\ket{\widehat{\bot}_\lambda},
\end{equation}
where $\ket{\widehat{\bot}_\lambda}$ denotes a state orthogonal to $\ket{0^{\otimes (m + 1)}, \psi_\lambda}$. Then
\begin{equation}
    \ket{\c_\lambda^\pm} = \frac{1}{\sqrt{2}}( \ket{0^{\otimes (m + 1)}, \psi_\lambda} \pm i\ket{\widehat{\bot}_\lambda})
\end{equation}
are eigenstates of $\widehat{U}_A$ with eigenphases $\pm \tau_\lambda = \pm\arccos\lambda$. Now the spectrum of $A$ and that of $\widehat{U}_A$ are associated, which allows us to process and extract the spectrum of an arbitrary matrix $A$ by applying QPP on $\widehat{U}_A$. More details of qubitization are discussed in Appendix~\ref{appendix:qubitization intro}. In the next section, we will show how to use QPP to solve the Hamiltonian simulation problem.

\subsection{Hamiltonian simulation}

Suppose we have a block encoding $U_H=[\begin{smallmatrix} H/\Lambda & \cdot \\ \cdot & \cdot \end{smallmatrix}]$. With out loss of generality, we may assume $\norm{H} \leq 1$ and $\L = 1$, since for $\norm{H} > 1$ the problem can be considered as simulating evolution under the rescaled Hamiltonian $H/\L$ for time $\L t$. Recall that the qubitization establishes the relation between the eigenvalues of the Hamiltonian $H$ and eigenphases of its qubitized block encoding $\widehat{U}_H$, i.e.\ $\tau_\lambda = \arccos(\lambda)$. Since the time-evolution operator $e^{-iHt}$ can be decomposed as $e^{-i\lambda t}$, the main idea is to transform eigenphases of $\widehat{U}_H$ by Theorem \ref{thm:qpp evolution} as $\tau_\lambda \mapsto e^{-i\lambda t}$. We select a trigonometric polynomial $F(x)$ that approximates the function $f(x)=e^{-i\cos(x) t}$ with desired precision. Then applying the \tp{} $F(x)$ on each eigenphase $\tau_\lambda$ approximates 
\begin{equation}
    f(\tau_\lambda) = e^{-i\cos(\tau_\lambda) t} = e^{-i\lambda t},
\end{equation}
which provides a precise approximation of the time-evolution operator $e^{-iHt}$. The query complexity of Hamiltonian simulation is analyzed in the following theorem, and the detailed proof is deferred to Appendix~\ref{appendix:hamiltonian simulation}.

\begin{theorem}\label{thm:ham simulation complexity}
Given a block encoding $U_H$ of $H/\Lambda$ for some $\Lambda \geq \norm{H}$, there exists an algorithm that simulates evolution under the Hamiltonian $H$ for time $t \in \mathbb{R}$ within precision $\delta>0$, using two ancilla qubits and querying controlled-$U_H$ and controlled-$U_H^\dagger$ for a total number of times in \[\Theta\left(\Lambda|t|+\frac{\log(2/\delta^2)}{\log\left(e+\frac{\log(2/\delta^2)}{\Lambda|t|}\right)}\right).\]
\end{theorem}


{For remark, one could measure the ancilla qubit after applying our circuit, and the residual state is highly approximate to the desired state.} In this case, we can relax the function approximation error. The final state approximates the target state with an error of $\delta$, succeeding with a probability at least $1-2\delta$. {As a result}, the circuit depth can be reduced.

From the result, we conclude that QPP solves Hamiltonian simulation problems with the access to block encoding, and the query complexity matches the optimal results as in~\cite{low2019hamiltonian}. We notice that QPP could also solve other Hamiltonian problems like spectrum estimation, ground state energy estimation, and ground state preparation. Details are in Appendix \ref{appendix:hamiltonian simulation}.


\section{Concluding remarks}


In this paper, we introduce a new quantum phase processing (QPP) algorithmic toolbox based on applying trigonometric transformations to eigenphases of a unitary operator. The toolbox allows us to implement a desired evolution to the input state and, more interestingly, to extract the eigen-information of quantum systems by simply measuring the ancilla qubit. We demonstrate the capability of this toolbox by developing QPP-based algorithms for solving a variety of problems. Owing to the capability of QPP to directly process eigenphases, we design an efficient and intuitive phase estimation algorithm purely based on the idea of binary search, which uses only one ancilla qubit and matches the best prior performance. Utilizing block encoding and qubitization, we show that QPP could reproduce and even improve previous quantum algorithms based on the framework of QSVT, such as Hamiltonian simulation and quantum entropies estimation.

QPP generalizes the trigonometric QSP by extending the $R_z$ rotation instead of the reflection in the Chebyshev-based QSP as QSVT did, providing a new interpretation for lifting QSP to multiple qubits. Due to the common underlying logic between these two QSP conventions, it is unavoidable to find similarities between QPP and QSVT. However, due to the distinctions between trigonometric and Chebyshev-based QSP, our results show that QPP is a powerful toolbox for developing quantum algorithms related to eigenphase transformation and processing, which essentially complements the existing QSVT framework. On one hand, QPP implements arbitrary complex trigonometric polynomial, which overcomes the parity constraints in QSVT and thus exempts the use of linear-combination-of-unitaries in certain cases. On the other hand, QPP is natively compatible with indirect measurements, which could extract eigen-information of the main system via measuring the single ancilla qubit. Notably, QPP could work without amplitude estimation, which requires shorter circuits and less coherence time than QSVT, and hence might be more friendly to near-term quantum hardware. 
QPP natively inherits the trick of phase kickback, and thus it is particularly suitable to develop phase-related quantum algorithms. Other than applications mentioned in this paper, QPP can be potentially applied to other problems, including but not limited to R\'enyi divergence estimation, unitary trace estimation, and quantum Monte Carlo method. Moreover, considering the connections between quantum signal processing and single-qubit quantum neural networks~\cite{perez-salinas2020data, yu2022power}, our results may shed lights and lead to  applications of QSVT and QPP in the area of quantum machine learning~\cite{Cerezo2022,Biamonte2017a,Li2022,Du2023,Wang2020d,Romero2017,Cerezo2020,Yu2022,Larocca2022,Liu2022c}. Overall, we believe that the QPP algorithmic toolbox would deepen our understanding of quantum algorithm design and shed light on the search for quantum applications in quantum physics, quantum chemistry, machine learning, and beyond.

\textbf{Acknowledgement.--}
Part of this work was done when Y. W.,  L. Z., Z. Y., and X. W. were at Baidu Research. X. W.  acknowledges support from the Start-up Fund from the Hong Kong University of Science and Technology (Guangzhou). 
Y. W. acknowledges support from the Baidu-UTS AI Meets Quantum project, the National Natural Science Foundation of China (No. 62071240), and the Innovation Program for Quantum Science and Technology (No. 2021ZD0302900).


%


\appendix

\vspace{2cm}
\onecolumngrid
\vspace{2cm}

\begin{center}
\textbf{
{\Large{Appendix}}}
\end{center}

\renewcommand{\theequation}{\thesection.\arabic{equation}}
\renewcommand{\thetheorem}{S\arabic{theorem}}
\renewcommand{\theproposition}{S\arabic{proposition}}
\renewcommand{\thelemma}{S\arabic{lemma}}
\renewcommand{\theremark}{S\arabic{remark}}
\renewcommand{\thecorollary}{S\arabic{corollary}}
\renewcommand{\thefigure}{S\arabic{figure}}
\renewcommand{\thetable}{S\arabic{table}}
\setcounter{equation}{0}
\setcounter{table}{0}
\setcounter{proposition}{0}
\setcounter{lemma}{0}
\setcounter{corollary}{0}
\setcounter{figure}{0}
\setcounter{remark}{0}
\setcounter{figure}{0}
\setcounter{table}{0}

In this supplementary material, we provide detailed analyses of results stated in this paper, and discuss some further applications of the QPP algorithmic toolbox.

\section{A brief summary of common QSP conventions}

\begin{table}[htbp]
\centering
\setlength{\tabcolsep}{1em}
\begin{tabular}{ccc}
\toprule
Conventions of QSP & QSP circuits & Type of polynomials \\
\midrule
\multirow{2}{*}{$W_x$~\cite{low2016methodology}}  & \multirow{2}{*}{\raisebox{-0.3\totalheight}{\includegraphics[width=0.26\textwidth]{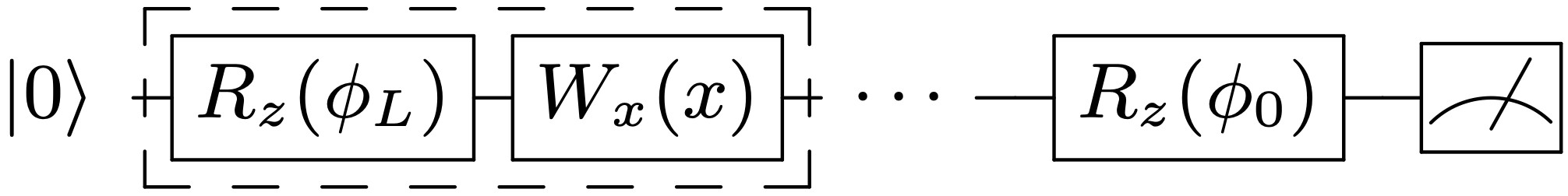}}} & \multirow{2}{*}{$\CC[x]$ with parity $L \bmod{2} $}  \\
\addlinespace
{} & {} & {and other constraints} \\
\midrule
\addlinespace
\multirow{2}{*}{Reflection~\cite{gilyen2019quantum}} & \multirow{2}{*}{\raisebox{-0.3\totalheight}{\includegraphics[width=0.25\textwidth]{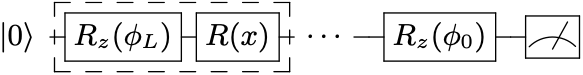}}} & \multirow{2}{*}{$\CC[x]$ with parity $L \bmod{2} $} \\
\addlinespace
{} & {} & {and other constraints} \\
\midrule
\addlinespace
\multirow{2}{*}{$W_z$~\cite{chao2020findinga, haah2019product}} &  \multirow{2}{*}{\raisebox{-0.3\totalheight}{\includegraphics[width=0.25\textwidth]{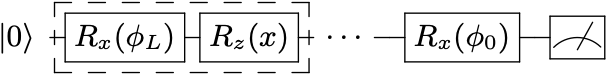}}} & \multirow{2}{*}{$\RR[e^{ix/2}, e^{-ix/2}]$ with parity $L \bmod{2} $} \\
\addlinespace
{or, trigonometric} & {} & {} \\
\midrule
\addlinespace
\multirow{2}{*}{Improved} & \multirow{2}{*}{\raisebox{-0.3\totalheight}{\includegraphics[width=0.42\textwidth]{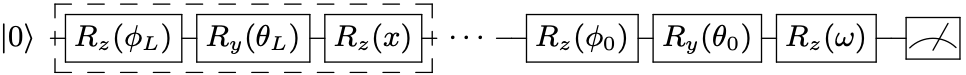}}} & \multirow{2}{*}{$\CC[e^{ix/2}, e^{-ix/2}]$ with parity $L \bmod{2} $}\\
\addlinespace
{trigonometric~\cite{yu2022power}} & {} & {} \\
\bottomrule
\end{tabular}
\caption{A table that summarizes the quantum circuits and the corresponding types of realizable polynomials retrieved by various conventions of QSP, where most of them were originally done by \citet{martyn2021grand} in their Appendix A. 
Note that $W_x(x) \coloneqq \begin{bmatrix} x & i\sqrt{1-x^2}\,\, \\ i\sqrt{1-x^2} & x \end{bmatrix}$, $R(x) \coloneqq \begin{bmatrix} x & \sqrt{1-x^2}\,\, \\ \sqrt{1-x^2} & -x \end{bmatrix}$, and all polynomials are assumed to be normalized.}
\label{tab:qsp circuits}
\end{table}

As shown in Table~\ref{tab:qsp circuits}, the trigonometric QSP can use a single-qubit quantum system to make arbitrary transformation of a polynomial in $\CC[e^{ix/2}, e^{-ix/2}]$ with parity. However, since all normalized and square-integrable functions can be approximated by the polynomials in $\CC[e^{ix/2}, e^{-ix/2}]$ with parity $0$ i.e.\ the polynomials in $\CC[e^{ix}, e^{-ix}]$, the trigonometric QSP can be claimed to be parity-free.

\section{Theorems of Quantum Phase Processing}

\subsection{Existence of Laurent complementary and angle finding }\label{appendix:angle finding}




Other than the characterization of trigonometric QSP in Lemma~\ref{lem:trig_qsp}, we also need to specify the achievable trigonometric polynomials $P(x)$ for which there exists a corresponding $Q(x)$ satisfying the three conditions in Lemma~\ref{lem:trig_qsp}. First we prove the following lemma, using similar ideas of root finding from previous works~\cite{haah2019product, chao2020findinga, silva2022fourierbased}.


\begin{lemma} \label{lem:sqrt existence}
    Suppose $A(x)$ is a non-negative real-valued \tp{}. Then there exists a Laurent polynomial $Q \in \CC[e^{ix/2}, e^{-ix/2}]$ such that $QQ^* = A$.
\end{lemma}
\begin{proof}
Denote $A(x) = \sum_{j=-2L}^{2L} a_j e^{i j x / 2}$. We can decompose $A$ by the set of $4L$ roots $\{\x_k\}_{k=1}^{4L}$ so that
\begin{equation}
    A(x) = a_{2L} e^{-i L x} \prod_{k=1}^{4L} (e^{ix / 2} - \x_k),
\end{equation}
where $\{\x_k\}_{k=1}^{4L}$ can be efficiently found by computing all roots of a regular complex polynomial $g(\x) \coloneqq \sum_{j = 0}^{4L} a_{j - 2L} \x^{j}$. In particular, since $A$ is real and non-negative, these roots appear in inverse conjugate pairs i.e.\ $\{\x_k\}_{k=1}^{4L} = \{\x_k, \frac{1}{\x_k^*}\}_{k=1}^{2L}$. Then $A$ can be further simplified to
\begin{equation}
A(x) = a_{2L}e^{-i L x}\left[\prod_{k=1}^{2L} (e^{ix / 2} - \x_k)\right] \left[\prod_{k=1}^{2L}(e^{ix / 2} - \frac{1}{\x_k^*})\right]
.\end{equation}
From the fact $e^{ix / 2} - \x_k = -e^{ix / 2}\x_k(e^{-ix /2} - \frac{1}{\x_k})$, we have
\begin{equation}
A(x) = a_{2L}\left[\prod_{k=1}^{2L} \x_k\right] \left[\prod_{k=1}^{2L} (e^{-ix / 2} - \frac{1}{\x_k})\right] \left[\prod_{k=1}^{2L}(e^{ix / 2} - \frac{1}{\x_k^*})\right]
.\end{equation}
Let $q \coloneqq a_{2L}\prod_{k=1}^{2L} \x_k$. Note that $q$ is real since 
\begin{equation}
A(0) = a_{2L}g(0) = a_{2L}\prod_{k=1}^{2L} \frac{\x_k}{\x_k^*} = \frac{a_{2L} (\prod_{k=1}^{2L} \x_k)^2}{\prod_{k=1}^{2L} |\x_k|^2}
\end{equation}
is real. Define $Q(x) \coloneqq \sqrt{q} e^{-i L x / 2} \prod_{k=1}^{2L}(e^{ix / 2} - \frac{1}{\x_k^*})$ and the result follows.
\end{proof}

Using Lemma~\ref{lem:sqrt existence}, we show that there exists a $Q$ such that $QQ^* = 1 - PP^*$ for any trigonometric series $P$ satisfying $|P|^2 \le 1$.

\begin{lemma}[Existence of Laurent complementary] \label{lem:Laurent complement}
Let $P \in \CC[e^{ix/2}, e^{-ix/2}]$ be a Laurent polynomial with degree no larger than $L$ and parity $L\bmod{2}$. Then $\abs{P(x)} \leq 1$ for all $x \in \RR$ if and only if there exists a Laurent polynomial $Q \in \CC[e^{ix/2}, e^{-ix/2}]$ satisfying
\begin{enumerate}
    \item $\deg(Q) \leq L$,
    \item $Q$ has parity $L\bmod{2}$,
    \item $\forall x\in \RR$, $|P(x)|^2 + |Q(x)|^2 = 1$.
\end{enumerate}
\end{lemma}
\begin{proof}
($\Longleftarrow$) The statement trivially holds from the third condition $|P(x)|^2 + |Q(x)|^2 = 1$.

($\Longrightarrow$) Suppose $P \in \CC[e^{ix/2}, e^{-ix/2}]$ is a Laurent polynomial satisfying above requirements. Note that if $\abs{P(x)} = 1$ for all $x \in \RR$, then the result follows by setting $Q = 0$. Suppose $\abs{P(x)} < 1$ for some $x \in \RR$. Then $A = 1 - PP^*$ is a real and positive Laurent polynomial. By Lemma~\ref{lem:sqrt existence}, there exists a Laurent polynomial $Q \in \CC[e^{ix/2}, e^{-ix/2}]$ such that $QQ^* = 1 - PP^*$ i.e.\ $|P(x)|^2 + |Q(x)|^2 = 1$ for all $x \in \RR$. The first and second conditions follow by the fact that $\deg(P) \leq L$ and $P$ has parity $L\bmod{2}$.
\end{proof}
\renewcommand{\thelemma}{S\arabic{lemma}}

Combined with Lemma~\ref{lem:Laurent complement} and Lemma~\ref{lem:trig_qsp}, now we can compute the corresponding rotation angles $\a$, $\bm\theta$ and $\bm\phi$ of $\UQSP$ for a given trigonometric polynomial $P(x)$, as shown in Algorithm~\ref{alg:angle finding}. 

\begin{figure}[ht]
\begin{algorithm}[H]
\caption{Angle Finding}
\label{alg:angle finding}
\begin{algorithmic}[1]
    \REQUIRE A Laurent polynomial $P \in \CC[e^{ix/2}, e^{-ix/2}]$ such that $\abs{P(x)} \leq 1$ for all $x \in \RR$.
    \ENSURE Rotation parameters $\omega, \bm\theta, \bm\phi$ such that $\bra{0}\UQSP(x)\ket{0} = P(x)$ for all $x \in [-\pi, \pi]$.

\STATE Compute the set of roots $\{\x_k\}_{k=1}^{4L}$ and leading coefficient $a_{2L}$ from real-valued and non-negative trigonometric polynomial $A(x) = 1 - P(x)P^*(x)$, where $L$ is the degree of $P$. Sort the set by ascending modulus and determine
\begin{equation}
    Q \leftarrow \sqrt{a_{2L}\prod_{k=1}^{2L}\x_k}e^{-i L x / 2}\prod_{k=1}^{2L}(e^{ix / 2} - \frac{1}{\x_k^*})
\end{equation}

\WHILE{$\deg(P) > 0$}

\STATE $k \leftarrow \deg(P)$, $p_k \leftarrow P[e^{i k x / 2}]$, $p_{-k} \leftarrow P[e^{-i k x / 2}]$, $q_k \leftarrow Q[e^{i k x / 2}]$ and $q_{-k} \leftarrow Q[e^{-i k x / 2}]$.

\STATE Determine $\theta_k, \phi_k$ such that $\cos{\frac{\theta_k}{2}}e^{i\frac{\phi_k}{2}}p_{k}+\sin{\frac{\theta_k}{2}}e^{-i\frac{\phi_k}{2}}q_k = 0$, $-\sin{\frac{\theta_k}{2}}e^{i\frac{\phi_k}{2}}p_{-k}+\cos{\frac{\theta_k}{2}}e^{-i\frac{\phi_k}{2}}q_{-k} = 0$.

\STATE Update polynomials (simultaneously) such that 
\begin{equation}
    P \leftarrow e^{i\frac{\phi_k}{2}}\cos\frac{\theta_k}{2}e^{ix / 2} P + e^{-i\frac{\phi_k}{2}}\sin\frac{\theta_k}{2}e^{ix / 2} Q, \quad Q \leftarrow e^{-i\frac{\phi_k}{2}}\cos\frac{\theta_k}{2}e^{-ix / 2} Q - e^{i\frac{\phi_k}{2}}\sin\frac{\theta_k}{2}e^{-ix / 2} P
\end{equation}
\ENDWHILE

\STATE Determine $\omega, \theta_0, \phi_0$ such that
$
R_z(\omega)R_y(\theta_0)R_z(\phi_0) =  
    \begin{bmatrix}
        P & -Q \\
        Q^* & P^*
    \end{bmatrix}
$; $\bm{\theta} \leftarrow (\theta_0, ..., \theta_L)$, $\bm{\phi} \leftarrow (\phi_0, ..., \phi_L)$.


\STATE Return $\omega$, $\bm\theta$ and $\bm\phi$.
\end{algorithmic}
\end{algorithm}
\end{figure}

\subsection{Proofs of Corollary~\ref{coro:trig_qsp_proj} and Corollary~\ref{coro:trig_qsp_Z_measurement}}\label{appendix:trig_qsp_outcome}
\renewcommand\thecorollary{\ref{coro:trig_qsp_proj}}
\setcounter{corollary}{\arabic{corollary}-1}
\begin{corollary}
For any complex-valued trigonometric polynomial $F(x) = \sum_{j=-L}^L c_j e^{ijx}$ with \highlight{$\abs{F(x)} \leq 1$ for all $x \in \RR$}, there exist $\omega \in \RR$ and $\bm\theta,\bm\phi \in \RR^{2L+1}$ such that for all $x\in \RR$,
\begin{equation}
    \bra{0} \UQSP^{2L}(x) \ket{0} = F(x)
.\end{equation}
\end{corollary}
\begin{proof}
Note that $F$ is a Laurent polynomial with degree no larger than $2L$. By Lemma~\ref{lem:Laurent complement}, there exists a Laurent polynomial $G \in \CC[e^{ix/2}, e^{-ix/2}]$ such that $\deg(G) \leq 2L$, $G$ has parity $0$ and $|F(x)|^2 + |G(x)|^2 = 1$ for all $x \in \RR$. By Lemma~\ref{lem:trig_qsp}, there exists $\omega \in \RR$, $\bm\theta \in \RR^{2L+1}$ and $\bm\phi \in \RR^{2L+1}$ such that
\begin{equation}
    \UQSP^{2L} = \begin{bmatrix}
        F & -G \\
        G^* & F^*
    \end{bmatrix}.
\end{equation}
The result directly follows.
\end{proof}

\renewcommand\thecorollary{\ref{coro:trig_qsp_Z_measurement}}
\setcounter{corollary}{\arabic{corollary}-1}
\begin{corollary}
For any real-valued trigonometric polynomial $F(x) = \sum_{j=-L}^L c_j e^{ijx}$ with \highlight{$\abs{F(x)} \leq 1$ for all $x \in \RR$}, there exist $\omega \in \RR$ and $\bm\theta,\bm\phi \in \RR^{L+1}$ such that for all $x\in \RR$,
\begin{equation}
    f_W(x) \coloneqq \bra{0} \UQSP^L(x)^\dagger Z \UQSP^L(x)\ket{0} = F(x)
.\end{equation}
\end{corollary}
\begin{proof}
Note that $(1 \pm F(x))/2$ are non-negative real-valued trigonometric polynomials. Then by Lemma~\ref{lem:sqrt existence} there exist two Laurent  polynomials $P, Q$ such that for all $x \in \RR$,
\begin{equation}
    P(x)P^*(x) = \frac{1 + F(x)}{2} \quad \textrm{and} \quad Q(x)Q^*(x) = \frac{1 - F(x)}{2}
.\end{equation}
Observe that $P$ and $Q$ satisfy the three conditions in Lemma~\ref{lem:trig_qsp}, thus there exists $\omega \in \RR$, $\bm\theta  \in \RR^{L+1}$ and $\bm\phi \in \RR^{L+1}$ such that
\begin{equation}
    \braket{0 | \UQSP^L(x)^\dagger Z \UQSP^L(x) | 0 }
     = P(x)P^*(x) - Q(x)Q^*(x) = F(x),
\end{equation}
for all $x\in \RR$.
\end{proof}
\renewcommand{\thecorollary}{S\arabic{corollary}}

\subsection{Proof of Lemma~\ref{lem:qpp eigenspace decomp}}\label{appendix:qpp eigenspace decomp}

\renewcommand\thelemma{\ref{lem:qpp eigenspace decomp}}
\setcounter{lemma}{\arabic{lemma}-1}
\begin{lemma} [Eigenspace Decomposition of QPP]
    Suppose $U$ is an $n$-qubit unitary with spectral decomposition 
    \begin{equation}
        U = \sum_{j=0}^{2^n - 1} e^{i \tau_j} \ketbra{\chi_j}{\chi_j}
    .\end{equation}
    For all $L \in \NN$, $\omega \in \RR$ and $\bm\theta,\bm\phi \in \RR^{L+1}$, we have
    \begin{equation}
        \UQPP^L(U) = \bigoplus_{j=0}^{2^n - 1} (e^{-i \tau_j/2})^{L\bmod{2}} \cdot \UQSP^L(\tau_j)_{\mathbb{B}_j}
    \end{equation}
    where $\mathbb{B}_j \coloneqq \{\ket{0, \chi_j}, \ket{1, \chi_j}\}$.
\end{lemma}

\begin{proof}
Observe that the decomposition of unitaries is
\begin{align}
    \begin{bmatrix}
        U^\dagger & 0 \\
        0 & I
    \end{bmatrix} & =
    \sum_{j=0}^{2^n-1} e^{- i \tau_j} \ketbra{0}{0} \otimes \ketbra{\chi_j}{\chi_j}
    +\sum_{j=0}^{2^n-1} \ketbra{1}{1}\otimes \ketbra{\chi_j}{\chi_j}\\
    & = \sum_{j=0}^{2^n-1} e^{- i \frac{\tau_j}{2}}
    \begin{bmatrix}
        e^{-i\frac{\tau_j}{2}} & 0 \\
        0 & e^{i \frac{\tau_j}{2}}
    \end{bmatrix}
    \otimes \ketbra{\chi_j}{\chi_j} = \sum_{j=0}^{2^n-1} e^{- i \frac{\tau_j}{2}} R_z(\tau_j) \otimes \ketbra{\chi_j}{\chi_j}\\
    & = \bigoplus_{j=0}^{2^n-1} e^{-i \frac{\tau_j}{2}} R_z(\tau_j)_{\mathbb{B}_j}, \text{\quad and} \\
    \begin{bmatrix}
        I & 0\\
        0 & U
    \end{bmatrix} & = 
    \sum_{j = 0}^{2^n - 1}  \ketbra{0}{0}\otimes \ketbra{\chi_j}{\chi_j} +
    \sum_{j = 0}^{2^n - 1} e^{i \tau_j} \ketbra{1}{1} \otimes \ketbra{\chi_j}{\chi_j} \\
    & = \sum_{j = 0}^{2^n - 1} e^{i \frac{\tau_j}{2}}
   \begin{bmatrix}
       e^{-i\frac{\tau_j}{2}} & 0 \\
       0 & e^{i \frac{\tau_j}{2}}
   \end{bmatrix}
   \otimes \ketbra{\chi_j}{\chi_j} = \sum_{j = 0}^{2^n - 1} e^{i \frac{\tau_j}{2}} R_{z}(\tau_j) \otimes \ketbra{\chi_j}{\chi_j}\\
   & = \bigoplus_{j=0}^{2^n-1} e^{i \frac{\tau_j}{2}} R_z(\tau_j)_{\mathbb{B}_j}
.\end{align}
In a similar manner, we can decompose $R_y$ and $R_z$ gates applied on the first qubit, such that for any $\zeta \in \mathbb{R}$,
\begin{equation}
\label{eqn: new gates eigenspace decomp}
\begin{aligned}
R_y^{(0)}(\zeta) \otimes I & = R_y^{(0)}(\zeta) \otimes \sum_{j = 0}^{2^n - 1} \ketbra{\chi_j}{\chi_j} = R_y^{(0)}(\zeta) \otimes \bigoplus_{j = 0}^{2^n - 1} {1}_{\{\ket{\chi_j}\}} = 
\bigoplus_{j = 0}^{2^n - 1} R_y(\zeta) _{\mathbb{B}_j}
\\
R_z^{(0)}(\zeta) \otimes I & = ... = 
\bigoplus_{j = 0}^{2^n - 1} R_z(\zeta) _{\mathbb{B}_j}
\end{aligned}
\end{equation}
For convenience $\mathbb{B}_j$s are omitted in the rest of the proof. From above equations, for all even $L \in \NN$ and $\omega \in \RR, \bm\theta, \bm\phi \in \mathbb{R}^{L + 1}$,

\begin{align}
    \UQPP^{L}(U)  
& = \bigoplus_{j = 0}^{2^n - 1} R_z(\omega) R_y(\theta_0) R_z(\phi_0)
\prod_{l = 1}^{L / 2} e^{- i \frac{\tau_j}{2}} R_z(\tau_j) R_y(\theta_{2l-1})R_z(\phi_{2l-1}) e^{i \frac{\tau_j}{2}} R_z(\tau_j) R_y(\theta_{2l})R_z(\phi_{2l}) \\
& = \bigoplus_{j = 0}^{2^n - 1} R_y(\theta_0) R_z(\phi_0) \prod_{l = 1}^L R_z(\tau_j) R_y(\theta_l)R_z(\phi_l)
= \bigoplus_{j = 0}^{2^n - 1} \UQSP^{L}(\tau_j)
.\end{align}
Similar statement holds for odd $L \in \NN$.
\end{proof}
\renewcommand{\thelemma}{S\arabic{lemma}}

\subsection{Proofs of Theorem~\ref{thm:qpp evolution} and Theorem~\ref{thm:qpp evaluation}}\label{appendix:qpp theorems}

\renewcommand\thetheorem{\ref{thm:qpp evolution}}
\setcounter{theorem}{\arabic{theorem}-1}
\begin{theorem}[Quantum phase evolution]
Given an $n$-qubit unitary $U = \sum_{j=0}^{2^n - 1} e^{i \tau_j} \ketbra{\chi_j}{\chi_j}$, for any trigonometric polynomial $F(x) = \sum_{j=-L}^L c_j e^{ijx}$ with $\abs{F(x)} \leq 1$ for all $x \in \RR$, there exist $\omega \in \RR$ and $\bm\theta,\bm\phi \in \RR^{2L+1}$ such that
\begin{equation}
    \UQPP^{2L}(U) = \begin{bmatrix}
        F(U) & \ldots \\
        \ldots & \ldots
    \end{bmatrix}
,\end{equation}
where $F(U) \coloneqq \sum_{j=0}^{2^n - 1} F(\tau_j) \ketbra{\chi_j}{\chi_j}$.
\end{theorem}
\begin{proof}
By Corollary~\ref{coro:trig_qsp_proj}, there exists $\omega \in \RR$, $\bm\theta  \in \RR^{2L+1}$ and $\bm\phi \in \RR^{2L+1}$ such that $\braket{0| \UQSP^{2L}(x) |0} = F(x)$. Then for such $\omega, \bm\theta$ and $\bm\phi$,
\begin{align}
    \left(\bra{0} \otimes I^{\otimes n}\right) \UQPP^{2L}(U)  \left(\ket{0} \otimes I^{\otimes n}\right)
    &= \sum_{j=0}^{2^n - 1} \left(\bra{0} \otimes I^{\otimes n}\right) \UQPP^{2L}(U)  \left(\ket{0} \otimes \ketbra{\chi_j}{\chi_j} \right) \\
    &= \sum_{j=0}^{2^n - 1} \left[\left(\bra{0} \otimes I^{\otimes n}\right) \UQPP^{2L}(U) \ket{0, \chi_j}\right] \bra{\chi_j} \\
    &= \sum_{j=0}^{2^n - 1} \bra{0}_{\mathbb{B}_j} \UQSP^{2L}(U)_{\mathbb{B}_j}(\tau_j) \ket{0}_{\mathbb{B}_j} \bra{\chi_j}  \qquad (\text{by Lemma~\ref{lem:qpp eigenspace decomp}}) \\
    &= \sum_{j=0}^{2^n - 1} F(\tau_j) \ketbra{\chi_j}{\chi_j} \\
    &= F(U)
,\end{align}
as required.
\end{proof}
\renewcommand{\thetheorem}{S\arabic{theorem}}

\renewcommand\thetheorem{\ref{thm:qpp evaluation}}
\setcounter{theorem}{\arabic{theorem}-1}
\begin{theorem}[Quantum phase evaluation]
Given an $n$-qubit unitary $U = \sum_{j=0}^{2^n - 1} e^{i \tau_j} \ketbra{\chi_j}{\chi_j}$ and an $n$-qubit quantum state $\rho$, for any real-valued trigonometric polynomial $F(x) = \sum_{j=-L}^L c_j e^{ijx}$ with \highlight{$\abs{F(x)} \leq 1$ for all $x \in \RR$}, there exist $\omega \in \RR$ and $\bm\theta,\bm\phi \in \RR^{L+1}$ such that $\widehat{\rho} = \UQPP(U) \left(\ketbra{0}{0} \otimes \rho\right)  \UQPP(U)^\dagger$ satisfies
\begin{equation}
    f_V(U) \coloneqq \tr \left[ \left(Z^{(0)}\otimes I^{\otimes n}\right) \cdot \widehat{\rho} \right] = \sum_{j=0}^{2^n - 1} p_{j} F(\tau_j),
\end{equation}
where $p_j = \bra{\chi_j} \rho \ket{\chi_j}$ and $Z^{(0)}$ is a Pauli-$Z$ observable acting on the first qubit.
\end{theorem}

\begin{proof}
We begin the proof by decomposing the observable $Z^{(0)} \otimes I$. Note that
\begin{align}
    Z^{(0)} \otimes I = Z^{(0)} \otimes \sum_{j = 0}^{2^n - 1} \ketbra{\chi_j}{\chi_j} = \bigoplus_{j = 0}^{2^n - 1} Z_{\mathbb{B}_j}
,\end{align}
where $\mathbb{B}_j = \{\ket{0, \chi_j}, \ket{1, \chi_j}\}$. By Corollary~\ref{coro:trig_qsp_Z_measurement}, there exist $\omega, \bm\theta$ and $\bm\phi$ such that $\bra{0} \UQSP^{L}(x)^\dagger Z \UQSP^{L}(x) \ket{0} = F(x)$, we have
\begin{align}
    & \quad \, \tr \left[ \left(Z^{(0)}\otimes I\right) \cdot \widehat{\rho} \right]
    = \tr \left[  \left(Z^{(0)}\otimes I\right) \UQPP^{L}(U) \left(\sum_{j, k = 0}^{2^n - 1} \braket{\chi_j | \rho | \chi_k} \ketbra{0, \chi_j}{0, \chi_k}\right) \UQPP^{L}(U)^\dagger \right] \\
    & = \sum_{j = 0}^{2^n - 1} p_j \tr \left[ \bra{0, \chi_j} \UQPP^{L}(U)^\dagger  \left(Z^{(0)}\otimes I\right) \UQPP^{L}(U) \ket{0, \chi_j} \right] \qquad (\text{where $p_j \coloneqq \braket{\chi_j | \rho |\chi_j}$}) \\
    & = \sum_{j = 0}^{2^n - 1} p_j \braket{0 |\UQSP^{L}(\tau_j)^\dagger Z \UQSP^{L}(\tau_j) | 0} \qquad (\text{by Lemma~\ref{lem:qpp eigenspace decomp}}) \\
    & = \sum_{j = 0}^{2^n - 1} p_j F(\tau_j).
\end{align}
\end{proof}
\renewcommand{\thetheorem}{S\arabic{theorem}}

\section{Detailed Analysis of Quantum Phase Search}\label{appendix:eigenphase search}

\subsection{Proof of Lemma~\ref{lem:phase classification}}\label{appendix:phase classification}
First we show that there exists a trigonometric polynomial that approximates the square wave function $\sqw(x) \coloneqq \sgn(\sin x)$, following the results of approximating the sign function by polynomials~\cite{gilyen2019quantum}.
\begin{lemma}[Trigonometric approximation of the square wave function] \label{lem:trig sign approx}
For any $\Delta, \eps \in (0, 1)$, there exists a trigonometric polynomial $F \in \CC[e^{-ix}, e^{ix}]$ of degree $L = \mathcal{O}(\frac{1}{\Delta}\log\frac{1}{\eps})$ such that
\begin{itemize}
    \item for all $x \in [-\pi, \pi]$, $\abs{F(x)} \leq 1$, and
    \item for all $x \in [-\pi + \Delta, -\Delta] \cup [\Delta, \pi-\Delta]$, $\abs{\sqw(x) - F(x)} \leq \eps$,
\end{itemize}
where $\sqw(x) \coloneqq \sgn(\sin x)$ is the square wave function.
\end{lemma}
\begin{proof}
By Lemma 25 in~\cite{gilyen2019quantum}, there exist a polynomial $P \in \RR[x]$ of degree $L = \mathcal{O}(\frac{1}{\delta}\log \frac{1}{\eps})$ such that $\abs{P(x)} \leq 1$ for all $x \in [-1, 1]$ and $\abs{P(x) - \sgn(x)} \leq \eps$ for all $x \in [-1, 1] \setminus (-\delta, \delta)$. Let $\delta = \sin(\Delta)$, write the polynomial $P$ in a Chebyshev form as $P(x) = \sum_{j=0}^L a_j T_j(x)$ for some $a_j \in \CC$, then by a change of variable,
\begin{equation}
    P(\sin x) = \sum_{j=0}^{L} \frac{(-i)^j a_j}{2} \left[ e^{ijx} + (-1)^j e^{-ijx} \right]
\end{equation}
is a \tp{} of degree $L = \mathcal{O}(\frac{1}{\delta}\log \frac{1}{\eps}) = \mathcal{O}(\frac{1}{\Delta}\log \frac{1}{\eps})$. Simply let $F(x) = P(\sin x)$, then we have $\abs{F(x)} \leq 1$ for all $x \in [-\pi, \pi]$ and $\abs{\sgn(\sin x) - F(x)} \leq \eps$ for all $x \in [-\pi + \Delta, -\Delta] \cup [\Delta, \pi-\Delta]$.

\end{proof}

Then we show that there exists a QPP that utilizes the approximated square wave function to classify phases on the interval $[-\pi + \Delta, -\Delta] \cup [\Delta, \pi-\Delta]$.
\renewcommand\thelemma{\ref{lem:phase classification}}
\setcounter{lemma}{\arabic{lemma}-1}
\begin{lemma}[Phase classification]
Given a unitary $U = \sum_{j=0}^{2^n - 1} e^{i \tau_j} \ketbra{\chi_j}{\chi_j}$, then for any $\Delta \in (0, \pi)$ and $\eps \in (0, 1)$, there exists a QPP circuit $V(U)$ of $L = \cO \left(\frac{1}{\Delta}\log\frac{1}{\eps}\right)$ layers such that
\begin{equation}
    V (U) \ket{0, \chi_k} = \begin{dcases}
      \sqrt{1 - \eps_k} \ket{0, \chi_k} + \sqrt{\eps_k} \ket{1, \chi_k} & \text{if $\tau_k \in [\Delta, \pi-\Delta)$,}\\
      \sqrt{\eps_k} \ket{0, \chi_k} + \sqrt{1 - \eps_k} \ket{1, \chi_k} & \text{if $\tau_k \in(-\pi + \Delta, -\Delta]$,}
    \end{dcases}
\end{equation}
for $0 \leq k < 2^n$, where $\eps_k \in (0, \eps)$.
\end{lemma}
\begin{proof}
By Lemma~\ref{lem:trig sign approx}, there exists a \tp{} $f(x)$ approximating the square wave function with order $L = \mathcal{O}(\frac{1}{\Delta}\log\frac{1}{\eps})$, such that $L$ is a multiple of $4$ and $\abs{f(x) - \text{sqw}(x)} < \eps$ for all $x \in (-\pi + \Delta, -\Delta] \cup [\Delta, \pi-\Delta)$. It follows from Lemma~\ref{lem:sqrt existence} that there exist \tps{} $P = \sqrt{\frac{1 + f(x)}{2}}$ and $Q = \sqrt{\frac{1 - f(x)}{2}}$. By Theorem~\ref{thm:qpp evolution}, there exist parameters $\omega \in \RR$, $\bm\theta,\bm\phi \in \RR^{L+1}$ such that
\begin{align}
    V(U) \ket{0, \chi_k}
    & = \sqrt{\frac{1 + f(\tau_k)}{2}} \ket{0, \chi_k} + \sqrt{\frac{1 - f(\tau_k)}{2}} \ket{1, \chi_k}
.\end{align}
Denote $\eps_k \coloneqq \frac{1}{2}|f(\tau_k) - \text{sqw}(\tau_k)|$. For $\tau_k \in (-\pi + \Delta, -\Delta]$,
\begin{equation}
    V(U) \ket{0, \chi_k} = \sqrt{\eps_k} \ket{0, \chi_k} + \sqrt{1 - \eps_k} \ket{1, \chi_k}
;\end{equation}
and for $\tau_k \in [\Delta, \pi-\Delta)$,
\begin{equation}
   V(U) \ket{0, \chi_k} = \sqrt{1 - \eps_k} \ket{0, \chi_k} + \sqrt{\eps_k} \ket{1, \chi_k}.
\end{equation}
Since $\abs{f(x) - \text{sqw}(x)} < \eps$ on the $(-\pi + \Delta, -\Delta] \cup [\Delta, \pi-\Delta)$, we have $\eps_k < \eps$ for each $\tau_k \in (-\pi + \Delta, -\Delta] \cup [\Delta, \pi-\Delta)$.

\end{proof}
\renewcommand{\thelemma}{S\arabic{lemma}}

\subsection{Phase interval search} \label{appendix:PIS}

The phase interval is the region containing the eigenphase of the input state. According to Lemma~\ref{lem:phase classification}, the main idea is to iteratively shrink the phase interval by the binary search method. At each iteration, we can decide the next subinterval, either $(-\pi + \Delta, -\Delta]$ or $[\Delta, \pi-\Delta)$, depending on the measurement result of the ancilla qubit. As $\Delta$ is small, the length of the interval reduces by nearly half, and thus the size of which would exponentially converge to $2 \Delta$. 

Formally, let $U$ be a unitary, $\Delta \in (0, \frac{1}{2})$ and $\eps \to 0$ (detailed analysis of success probability is discussed in section~\ref{appendix:PS through IA}). Suppose $\UQPPs(.)$ is the quantum circuit stated in Lemma~\ref{lem:phase classification} with respect to $\Delta$ and $\eps$. Denote $\cG \coloneqq [-\pi, -\pi + \Delta] \cup [-\Delta, \Delta] \cup [\pi - \Delta, \pi]$ as the area that produces garbage information. Then for any quantum state $\ket{\psi} = \sum_{j=0}^{2^n - 1} c_j \ket{\chi_j}$ and $\z \in [-\pi, \pi]$, we have
\begin{align}
    \UQPPs(e^{-i\z} U)\ket{0, \psi}
    & = \ket{0}\left(\sum_{j:\Delta< \tau_j-\z <\pi-\Delta} c_j \ket{\chi_j} + \sum_{j: \tau_j-\z \in \cG} g_j^{(0)} \ket{\chi_j}\right) \nonumber \\
     & + \ket{1}\left(\sum_{j:-\pi+\Delta< \tau_j-\z <-\Delta} c_j \ket{\chi_j} + \sum_{j: \tau_j-\z \in \cG} g_j^{(1)} \ket{\chi_j}\right)
,\end{align}
where $g_j^{(i)}$ are garbage coefficients corresponding to state $\ket{i}$ of the ancilla qubit. Consequently, the measurement of the ancilla qubit can identify the interval containing the remaining eigenphases of measured state. Note that above approach is no longer applicable if the length of interval is close to $2 \Delta$, in which case $\tau_j-\z$ will always fall into the garbage area $\cG$. Such interval is referred to be ``indistinguishable''. We could iteratively apply the binary search procedure to shrink the phase interval until it becomes indistinguishable. The following corollary guarantees that Algorithm~\ref{alg:interval search} can reduce the length of input interval close to $2\Delta$ with high probability.


\begin{lemma}\label{lem:phase interval search}
    Suppose $\Delta, \eps$ and $\cQ$ are inputs of Algorithm~\ref{alg:interval search}, then the algorithm output an interval $[\z_l, \z_r]$ such that $\abs{\z_r- \z_l} = 2(\Delta+\frac{\pi}{2^{\cQ + 1}})$ and $\tau \in [\z_l, \z_r]$ with probability at least $(1-\eps)^{\cQ}$.
\end{lemma}
\begin{proof}
Denote $[\z_l^{(j)}, \z_r^{(j)}]$ as the interval generated at the end of the $j$-th iteration and $\z_m^{(j)}$ as the middle point of this interval. Observe that for $j > 0$, the interval generated at the end of the $j+1$-th iteration is either $(\z_m^{(j)}-\Delta, \z_r^{(j)})$ or $(\z_l^{(j)}, \z_m^{(j)}+\Delta)$. By induction we have
\begin{align}
    \abs{\z_r^{(\cQ)}-\z_l^{(\cQ)}}
& = \Delta + \frac{|\z_r^{(\cQ-1)}-\z_l^{(\cQ-1)}|}{2} = \Delta + \frac{1}{2}\left(\Delta + \frac{|\z_r^{(\cQ-2)} - \z_l^{(Q-2)}|}{2}\right) \\
& = \Delta + \frac{\Delta}{2} + \frac{\Delta}{2^2} + \ldots + \frac{\Delta}{2^{\cQ-1}}+\frac{|\z_r^{(0)} - \z_l^{(0)}|}{2^{\cQ}}\\
& = \Delta\left(\frac{1-2^{-\cQ}}{1-\frac{1}{2}}\right)+\frac{2\Delta+\pi}{2^{\cQ}}\\
&= 2\Delta + \frac{\pi}{2^\cQ}.
\end{align}
as required. By Lemma~\ref{lem:phase classification}, the probability of failing to decide the correct subinterval is at most $\eps$ in each iteration, thus the success probability of outputting an phase interval containing $\tau$ is at least $(1-\eps)^\cQ$.
\end{proof}

\subsection{Phase search through interval amplification} \label{appendix:PS through IA}

Let $\bar{\Delta} \coloneqq \Delta+\frac{\pi}{2^{\cQ + 1}}$, merely applying Algorithm~\ref{alg:interval search} will not provide an estimate within expected precision, given that the error is at most the $2\bar{\Delta}$. To address this issue, the main idea is to execute the phase interval search procedure on $(e^{i\z}U)^d$ for some appropriate integer $d$ so that the binary search procedure can continue to locate the amplified phase $d\tau \in [d\z_l, d\z_r]$, since the interval length is $2d\bar{\Delta} \gg 2\Delta$. Repeating the entire procedure can exponentially reduce estimation error.

Let us formally describe the entire procedure. In the first round of phase interval search, the initial phase interval is $[-\pi, \pi]$. We iteratively apply QPP $V(\cdot)$ on the target unitary $U^{(0)}=U$ to obtain a phase interval $[\z_l^{(0)}, \z_r^{(0)}]$ of length $2\bar{\Delta}$. We denote $\z_m^{(0)} \coloneqq (\z_l^{(0)} + \z_r^{(0)})/2$ the middle point of the interval. Let $d \coloneqq \lfloor 1 / \bar{\Delta} \rfloor$, we then construct a unitary as $U^{(1)}=(e^{-i \z_m^{(0)}} U^{(0)})^d$ so that the interval the amplified phase $d\tau$ are rescaled to $[-1, 1]$. Therefore, we can run the phase interval search procedure on $U^{(1)}$ to retrieve a new interval $[\z_l^{(1)}, \z_r^{(1)}]$ of length $2\bar{\Delta}$. Repeating the entire procedure above gives an estimation of phase $\tau$ up to required precision.



Now we analyse how above procedures improve the eigenphase estimation precision. For the target eigenphase $\tau$ of the input unitary $U$, the corresponding phase of unitary $U^{(1)}$ is $d(\tau - \z_m^{(0)}) \in [\z_l^{(1)},\z_r^{(1)}]$. Let $\z_m^{(1)} = (\z_l^{(1)}+\z_r^{(1)})/2$ denote the middle point of the second interval $[\z_l^{(1)},\z_r^{(1)}]$. Then we can readily give an inequality that characterizes the estimation error,
\begin{equation}
    \left|\z_m^{(1)}-d(\tau-\z_m^{(0)})\right|\leq \bar{\Delta}.
\end{equation}
We further rewrite this inequality as below,
\begin{equation}
    \left|\tau - \left(\z_m^{(0)} + \frac{\z_m^{(1)}}{d}\right) \right| \leq \frac{\bar{\Delta}}{d} \leq d^{-2}.
\end{equation}
From this equation, we can see that this scheme give an estimate of eigenphase with an error of $d^{-2}$. After repeating the procedure for sufficiently many times, we inductively obtain a sequence $(\z_m^{(0)}, \z_m^{(1)},\ldots)$ that could be taken as an estimate of the eigenphase. Therefore, the estimation error will exponentially decay as iterating. For instance, assuming our scheme executes $\cT$ times, it will inductively give an estimate as follows.
\begin{equation}
    \left|\tau - \left(\z_m^{(0)} + \frac{\z_m^{(1)}}{d} + \frac{\z_m^{(2)}}{d^2} + \cdots + \frac{\z_m^{(\cT - 1)}}{d^{\cT - 1}}\right) \right|\leq \frac{\bar{\Delta}}{d^{\cT-1}}\leq d^{-\cT}
.\end{equation} 
Note that $d = [1 / \bar{\Delta}]$ must be at least $2$, otherwise we cannot find a sequence that converges to the eigenphase.
Also, if the estimation precision is expected to be $\delta$, then $\cT$ should satisfy $d^{-\cT} \leq \delta$. Then it derives that
\begin{equation} \label{eqn:phase cT cQ bound}
    \cQ = \ceil*{\log\left(\frac{2\pi}{1 - 2\Delta}\right)} \quad \textrm{and} \quad \cT = \ceil*{\frac{\log(\frac{1}{\delta})}{\log(d)} }
.\end{equation}
The entire phase search procedure executes the phase interval search, i.e.\ Algorithm~\ref{alg:interval search}, for $\cT$ times. Then by Lemma~\ref{lem:phase interval search}, the success probability of the entire phase search procedure is $(1-\eta)^{\cQ\cT}$, where $\eta$ is the input threshold of Algorithm~\ref{alg:interval search}. For any $\eps \in (0,1)$, let $\eta = \frac{\eps}{\cQ \cT}$, then the total success probability is
\begin{equation}
    (1 - \eta)^{\cQ \cT} > 1 - \cQ \cT \eta  = 1 - \eps,
\end{equation}
where the first strict inequality is the Bernoulli's inequality for $\eta < 1$ and $\cQ \cT \ge 2$. By Lemma~\ref{lem:phase classification}, to ensure the success probability of the entire algorithm is at least $1 - \eps$, the number of QPP layers $L$ should be
\begin{equation} \label{eqn:layer L bound}
    L = \cO \left( \frac{1}{\Delta} \log\frac{\cQ \cT}{\eps} \right)
    = \cO \left( \frac{1}{\Delta} \log\frac{\log \frac{1}{\delta}}{\eps} + g(\Delta) \right)
.\end{equation}
Here $g(\Delta)$ is a function in terms of $\Delta$ only, and hence can be omitted in the complexity analysis of Theorem~\ref{thm:qps complexity}.

Overall, by running the scheme many times, we can find an estimate of eigenphase with precision $\delta$ and success probability at least $1 - \eps$. Above results are summarized in Algorithm~\ref{alg:eigenphase search}.


\subsection{Proof of Theorem~\ref{thm:qps complexity}} \label{appendix:qps complexity}

\renewcommand\thetheorem{\ref{thm:qps complexity}}
\setcounter{theorem}{\arabic{theorem}-1}
\begin{theorem} [Complexity of Quantum Phase Search]
    Given an $n$-qubit unitary $U$ and an eigenstate $\ket{\c}$ of $U$ with eigenvalue $e^{i\tau}$,  Algorithm~\ref{alg:eigenphase search} can use one ancilla qubit and $\cO \left(\frac{1}{\delta} \log \left(\frac{1}{\eps} \log\frac{1}{\d} \right) \right)$ queries to controlled-$U$ and its inverse to obtain an estimation of $\tau$ up to $\d$ precision with probability at least $1 - \eps$.
\end{theorem}
\begin{proof}
We analyze the number of queries to the controlled-$U$ oracle in Algorithm~\ref{alg:eigenphase search} to get an estimate within required precision $\d > 0$, probability of failure $\eps > 0$ and input $\Delta$. Note that $\Delta$ is a self-adjusted parameter and hence can be considered as a constant in complexity analysis.

Observe that Algorithm~\ref{alg:eigenphase search} executes the phase interval search procedure $\cT$ times, while the $t$-th phase interval search procedure requires $\cQ$ calls of circuit $\UQPPs(U^{d^t})$ of $L$ layers.
By Eq.~\eqref{eqn:phase cT cQ bound} and Eq.~\eqref{eqn:layer L bound}, the total query complexity of the controlled-$U$ oracle is 
\begin{equation}
    \sum_{t = 0}^{\cT - 1} \cQ \times d^t \times L  
    = \cO (d^\cT L \cQ ) = \cO \left(\frac{1}{\delta} \frac{1}{\Delta} \log(\frac{\log \frac{1}{\delta}}{\eps}) \log\left(\frac{2\pi}{1 - 2\Delta}\right) \right)
    = \cO \left(\frac{1}{\delta} \log \left(\frac{1}{\eps} \log \frac{1}{\delta} \right) \right).
\end{equation}
\end{proof}
\renewcommand{\thetheorem}{S\arabic{theorem}}

\subsection{Application: period finding and factoring} \label{appendix:app period finding}
In this section, we consider applying the quantum phase search algorithm to solve the period-finding problem. The goal of period-finding is to find the smallest integer $r$ (namely the order) of a given element $x$ in the rings of integers modulo $N \in \NN$ such that $x^{r} \equiv 1 \pmod N$.

In the quantum setting, the problem of reversible quantum modular multiplier has been well-studied~\cite{rines2018high, cho2020quantum}. It has been proved that the quantum operator
\begin{align}
    U_x\ket{y \bmod N} = \ket{xy \bmod N}, \, y \in \ZZ/N\ZZ
\end{align}
can be constructed in cubic resources for every integer $x \in \ZZ/N\ZZ$. A novel property of such modular multiplier is that there is no additional quantum cost to realize a power of $U_x$ since $U_x^j = U_{x^j}$. Moreover, an eigenvector of $U_x$ is in form
\begin{align}
    \ket{u_s} \coloneqq \frac{1}{\sqrt{r}}\sum_{k=0}^{r-1} \exp\left(\frac{2\pi i sk}{r}\right)\ket{x^k \bmod N}
\end{align}
corresponding to its eigenphase $\frac{2\pi s}{r}$, and the uniform superposition of all eigenvectors is $\ket{1}$, i.e.\ $\ket{1}=\frac{1}{\sqrt{r}}\sum_{s=0}^{r-1}\ket{u_s}$.

The conventional quantum period-finding algorithm is to apply the quantum phase estimation algorithm to the modular multiplier with input state $\ket{1}$, then use the continued fraction algorithm to extract the order from the estimated phases. Similarly, we use the quantum phase search algorithm to extract eigenphases of the modular operator. The details of the algorithm are shown below.
\begin{figure}[H]
\begin{algorithm}[H]
\caption{Quantum Period Finding Algorithm}
\label{alg:period_finding}
\begin{algorithmic}[1]
    \REQUIRE $N \in \ZZ$, $x \in \ZZ/N\ZZ$, constant $\Delta\in(0,\frac{1}{2})$ and error tolerance $\eps, \d > 0$.
    \ENSURE order $r$ of $x$ in $\ZZ/N\ZZ$.
\STATE Construct the modular multiplier $U_x$ by $x$ and $N$.
\STATE Retrieve an estimated eigenphase $\tau$ of $U_x$ by the $\QPS$ algorithm. That is, $\tau \leftarrow \QPS(U_x, \ket{1}, \Delta, \eps, \d)$.
\STATE Apply continued fractional algorithm on $\tau$ to retrieve $l$ and $r$; return $r$.
\end{algorithmic}
\end{algorithm}
\end{figure}

\subsection{Application: amplitude estimation} \label{appendix:app qae}
The problem of quantum amplitude estimation (QAE) can be efficiently solved by the phase estimation algorithm~\cite{brassard2002quantum}, providing a quadratic speedup over classical Monte Carlo methods. In recent years, several studies~\cite{suzuki2020amplitude, grinko2021iterative, aaronson2021quantum, rall2022amplitude} can realize phase-estimation-free amplitude estimation with same quantum speedup. However, these works require large number of samplings i.e.\ measurements from quantum circuits. Here we show that our phase search algorithm can also apply to the amplitude estimation and inherit the computational advantage from conventional phase estimation.

Let $\mathcal{A}$ denote a quantum circuit that acts on $n$ qubits. Applying the circuit $\mathcal{A}$ to $\ket{0^n}$, the produced state is of the following form:
\begin{align}
    \mathcal{A}\ket{0^{\otimes n}}=\cos(\tau)\ket{0}\ket{\psi}+\sin(\tau)\ket{1}\ket{\phi}. \label{eqn:amplitude}
\end{align}
where $\ket{\psi}$ and $\ket{\phi}$ are $(n-1)$-qubit states, and $\tau\in(-\pi,\pi)$ is the phase. Here, $\cos(\tau)$ and $\sin(\tau)$ denote the amplitude of states $\ket{\psi}$ and $\ket{\phi}$, respectively. Our goal is to estimate $|\sin(\tau)|$ up to a given precision with high probability.

Suppose we can repeatedly apply the circuit $\mathcal{A}$ and its inverse $\mathcal{A}^{\dagger}$. Then we can construct a circuit for the Grover operator
\begin{align} \label{eqn:grover op}
G=\mathcal{A}(2\ket{0^{\otimes n}}\bra{0^{\otimes n}} - I^{\otimes n})\mathcal{A}^\dagger\cdot (I-2\ketbra{1}{1}) \otimes I^{\otimes (n-1)}.
\end{align}
Note that, in the space spanned by $\{\ket{0}\ket{\psi}, \ket{1}\ket{\phi}\}$, the Grover operator $G$ has an eigenphase $2\tau$ or $-2\tau$. Thus our amplitude estimation algorithm just applies the quantum phase search algorithm to Grover operator. Specifically, applying our quantum phase search algorithm can effectively extract eigenphases $\pm2\tau$. Moreover, post-processing can estimate the amplitude within the required precision. We show more details below.
\begin{figure}[H]
\begin{algorithm}[H]
\caption{Quantum Amplitude Estimation Algorithm}
\label{alg:amplitude_estimation}
\begin{algorithmic}[1]
    \REQUIRE circuit $\mathcal{A}$, constant $\Delta\in(0,\frac{1}{2})$ and error tolerance $\eps, \d > 0$.
    \ENSURE an amplitude $\sin(\tau)$.
\STATE Prepare the initial state $\ket{\chi} = \mathcal{A}\ket{0}^{\otimes n}$ and construct the Grover operator $G$ in equation~\eqref{eqn:grover op}.
\STATE Retrieve an estimated eigenphase $\widetilde{\tau}$ of $G$ by the $\QPS$ algorithm. That is, $\widetilde{\tau} \leftarrow \QPS(G, \ket{\chi}, \Delta, \eps, \d)$.
\STATE Return $|\sin(\frac{\widetilde{\tau}}{2})|$.
\end{algorithmic}
\end{algorithm}
\end{figure}



\section{Further Review for Block Encoding}

\subsection{Qubitization} \label{appendix:qubitization intro}

In this section we review the technique of qubitization purposed by~\cite{low2019hamiltonian}. Such technique associates the spectrum of target block encoding $U_A$ with the block encoded matrix $A$. Assume $A$ is Hermitian with $\norm{A} \leq 1$, since our work only deal with Hamiltonians and density operators. Recent work has discussed an explicit construction scheme for building a block encoding sparse matrices \cite{camps2022explicit}. To better understand qubitization, we analyze the spectral information of the circuit in Fig.~\ref{fig:qubitization}.

\begin{figure}[H]
\[ 
\Qcircuit @C=1.5em @R=1.5em {
\lstick{\ket{0}} & /^1 \qw & \gate{H} & \ctrlo{1} & \ctrl{1} & \gate{X} & \gate{H} & \multigate{1}{\rm REFLECTOR} & \qw  \\
\lstick{\ket{0^{\otimes m}}} & /^m \qw & \qw & \multigate{1}{U_A} & \multigate{1}{U_A^\dagger} & \qw & \qw & \ghost{\rm REFLECTOR} & \qw \\
\lstick{} & /^n \qw & \qw & \ghost{U_A} & \ghost{U_A^\dagger} & \qw & \qw & \qw & \qw
\gategroup{1}{3}{3}{7}{1.5em}{--}
}
\]
\caption{Circuit realization for the qubitized unitary $\widehat{U}_A$. Here $U_A$ is a $(n + m)$-qubit block encoding of $n$-qubit matrix $A$, and the gate \textrm{REFLECTOR} is equivalent to $2\ketbra{0^{\otimes (m + 1)}}{0^{\otimes (m + 1)}} - I^{\otimes (m + 1)}$. Note that it suffices to control the \textrm{REFLECTOR} and three \textrm{X} gates to realize controlled-$\widehat{U}_A$.}
\label{fig:qubitization}
\end{figure}
Let $\widehat{U}_A$ denotes the qubitization of block encoding unitary $U_A$ and $\widetilde{U}_A$ denotes the unitary for dashed region in Fig.~\ref{fig:qubitization}. Then Lemma 10 of~\cite{low2019hamiltonian} implies that $\widetilde{U}_A$ satisfies
\begin{align}
    (\bra{0^{\otimes (m + 1)}} \otimes I^{\otimes n})\, \widetilde{U}_A \,(\ket{0^{\otimes (m + 1)}} \otimes I^{\otimes n}) & = A, \\
    (\bra{0^{\otimes (m + 1)}} \otimes I^{\otimes n})\, \widetilde{U}_A^2 \,(\ket{0^{\otimes (m + 1)}} \otimes I^{\otimes n}) & = I^{\otimes n}
    \label{eqn:qubitization_condition}
.\end{align}
Let $\ket{\psi_\lambda}$ be an eigenstate of $A$ corresponding to its eigenvalue $\lambda$. Denote 
$\ket{\widehat{\psi}_\lambda} \coloneqq \ket{0^{\otimes (m + 1)}, \psi_\lambda}$. After applying $\widetilde{U}_A$ and $\widetilde{U}_A^2$ to $\ket{\widehat{\psi}_\lambda}$ respectively, states are of the following form:
\begin{align}
    \widetilde{U}_A \ket{\widehat{\psi}_\lambda} &= \lambda \ket{\widehat{\psi}_\lambda} +\sqrt{1-\lambda^2}\ket{\widehat{\bot}_\lambda}, \\
    \widetilde{U}_A^2 \ket{\widehat{\psi}_\lambda} &= \ket{\widehat{\psi}_\lambda},
    \label{eqn:u_bot}
\end{align}
where $\ket{\widehat{\bot}_\lambda}$ is an orthogonal state and satisfies $(\ketbra{0^{\otimes (m + 1)}}{0^{\otimes (m + 1)}} \otimes I^{\otimes n})\ket{\widehat{\bot}_\lambda}=0$. Above results also imply that all subspaces ${\rm span}\{\ket{\widehat{\psi}_\lambda}, \ket{\widehat{\bot}_\lambda}\}$ are mutually perpendicular. 
Moreover, Eq.~\eqref{eqn:u_bot} implies
\begin{equation}
    \widetilde{U}_A\ket{\widehat{\bot}_\lambda} = \sqrt{1-\lambda^2}\ket{\widehat{\psi}_\lambda} - \lambda\ket{\widehat{\bot}_\lambda}
.\end{equation}
Note that it suffices to analyze $\widehat{U}_A$ under subspace 
\begin{equation} \label{eqn:hilbert decomposition}
    \cH_A \coloneqq \bigoplus_{\lambda} {\rm span}\{\ket{\widehat{\psi}_\lambda}, \ket{\widehat{\bot}_\lambda}\}
,\end{equation}
since the input state of ancilla qubits is always $\ket{0^{\otimes (m + 1)}}$. In this subspace, we can see that $\widehat{U}_A$ is essentially a rotation whose matrix is similar to $R_Y$, i.e.\
\begin{align}
  \widehat{U}_A = (\textrm{REFLECTOR} \otimes I^{\otimes n}) \cdot \widetilde{U}_A = \bigoplus_{\lambda}
  \begin{bmatrix}
      \lambda & -\sqrt{1-\lambda^2}\\
      \sqrt{1-\lambda^2} & \lambda
  \end{bmatrix}_{\{\ket{\widehat{\psi}_\lambda},  \ket{\widehat{\bot}_\lambda}\}} \oplus [\cdots]_{\cH_A^\bot}.
  \label{eqn:rotation}
\end{align}
The spectral details of $\widehat{U}_A$ in $\cH$ follow immediately:
\begin{align}
\begin{matrix}
   \textrm{ eigenvector} & \ket{\chi_\lambda^{+}}=\frac{1}{\sqrt{2}}(\ket{\widehat{\psi}_\lambda} + i\ket{\widehat{\bot}_\lambda}), & \textrm{ eigenvalue $e^{+i\tau_\lambda}$, where $\lambda=\cos(\tau_\lambda)$};\\
    \textrm{ eigenvector} &\ket{\chi_\lambda^{-}}=\frac{1}{\sqrt{2}}(\ket{\widehat{\psi}_\lambda} - i\ket{\widehat{\bot}_\lambda}), & \textrm{ eigenvalue $e^{-i\tau_\lambda}$, where $\lambda=\cos(-\tau_\lambda)$}.
\end{matrix}
\label{eqn:spectral}
\end{align}
Therefore, we can select $\widehat{U}_A$ as the input unitary in the QPP circuit, to access the arc-cosine of eigenvalues of $A$, allowing phase evolution and evaluation to be applied on block encoded matrices.

\subsection{Block encoding construction for density matrices} \label{appendix:dm block enc}

The quantum purification model, which prepares a purification of a mixed state $\rho$, is an extensively explored model for entropy in the literature. Consider quantum registers $A$ and $B$ storing $n$ and $n'$ qubits, respectively. Suppose we have accessed to a $(n + n')$-qubit unitary oracle $U_\rho$ acting on these two registers, such that
\begin{equation}
    \ket{\Psi}_{AB} \coloneqq U_\rho \ket{0^{\otimes n}}_A \ket{0^{\otimes n'}}_B \textrm{ and } \tr_{B}(\ket{\Psi}_{AB}\bra{\Psi}_{AB}) = \rho
.\end{equation}
Such oracle can be further employed to construct a block encoding $\widehat{U}_\rho$ of $\rho$. We recall Lemma 7 in \cite{low2019hamiltonian} and give the circuit construction for $\widehat{U}_\rho$ in Fig.~\ref{fig:state_block encoding}, using $U_\rho$ and $U_\rho^\dagger$ once.
\begin{figure}[H]
    \[ 
    \Qcircuit @C=1.5em @R=1.5em {
    \lstick{\ket{0^{\otimes n'}}_{B}} & /^{n'} \qw & \multigate{1}{U_\rho} & \qw & \multigate{1}{U_\rho^\dagger} & \multigate{1}{\rm REFLECTOR} & \qw \\
    \lstick{\ket{0^{\otimes n}}_{A}} & /^n \qw& \ghost{U_\rho} & \multigate{1}{\rm SWAP} & \ghost{U_\rho^\dagger} & \ghost{\rm REFLECTOR} & \qw \\
    \lstick{} & /^n \qw & \qw & \ghost{\rm SWAP} & \qw & \qw & \qw^{\rm} 
    \gategroup{1}{3}{3}{5}{1.5em}{--}
    } \] 
\caption{Circuit realization for $\widehat{U}_\rho$. The input state for two ancilla registers is $\ket{0^{\otimes (n + n')}}_{AB}$. Here \textrm{SWAP} is a $2n$-qubit swap operator and \textrm{REFLECTOR} is equivalent to $2\ketbra{0^{\otimes (n + n')}}{0^{\otimes (n + n')}} - I^{\otimes (n + n')}$. Note that it suffices to control the \textrm{SWAP} and \textrm{REFLECTOR} gates to realize controlled-$\widehat{U}_\rho$.}
\label{fig:state_block encoding}
\end{figure}

Note that the unitary $\widetilde{U}_\rho$ for the dashed region in Fig.~\ref{fig:state_block encoding} satisfies
\begin{align}
    (\bra{0^{\otimes (n + n')}}_{AB} \otimes I^{\otimes n})\, \widetilde{U}_\rho \,(\ket{0^{\otimes (n + n')}}_{AB} \otimes I^{\otimes n}) & = \rho, \\
    (\bra{0^{\otimes (n + n')}}_{AB} \otimes I^{\otimes n})\, \widetilde{U}_\rho^2 \,(\ket{0^{\otimes (n + n')}}_{AB} \otimes I^{\otimes n})  & = I^{\otimes n}
.\end{align} 
Denote $\{ p_j \}_{j=0}^{2^n - 1}$ as the set of eigenvalues of $\rho$, and $\{ \tau_j \}_{j=0}^{2^{n + 1} - 1}$ as the set of eigenphases of $\widehat{U}_\rho$ under subspace $\cH_\rho$ defined in Eq.~\eqref{eqn:hilbert decomposition}. Through same reasoning in Appendix~\ref{appendix:qubitization intro}, we have
\begin{equation} \label{eqn:phase cosine relation}
    \{ \tau_j \}_{j=0}^{2^{n + 1} - 1} = \{ \arccos(p_j), -\arccos(p_j) \}_{j=0}^{2^n - 1}
.\end{equation}
As shown above, the spectrum of $\widehat{U}_\rho$ is connected with that of $\rho$ in subspace $\cH_\rho$.

\section{Further Discussion for Hamiltonian Problems} \label{appendix:hamiltonian simulation}
\subsection{Hamiltonian simulation}
In this section, we explain the main idea of our method for Hamiltonian simulation and discuss how to find parameters for simulating the target function $f(x)$. For convenience, let $\widehat{U}_H$ denote the qubitized block encoding of a Hamiltonian $H$.

Prepare the initial state $\ket{0^{\otimes (m + 2)}, \psi}$, where the first one is the ancilla qubit of QPP, and the other $m+1$ qubits are ancilla qubits of the qubitized block encoding $\widehat{U}_H$. Decompose the initial state $\ket{0^{\otimes (m + 2)},\psi}$ by eigenvectors of the Hamiltonian. 
\begin{align}
    \ket{0}\ket{0^{\otimes (m + 1)},\psi}=\sum_{\lambda} \beta_\lambda \ket{0}\ket{0^{\otimes (m + 1)},\psi_\lambda}.
\end{align}
where $\sum_{\lambda}|\beta_\lambda|^2=1$. As shown in Eq.~\eqref{eqn:rotation}, the qubitized block encoding $\widehat{U}_H$ is a rotation in each subspace ${\rm span}\{\ket{0^{\otimes (m + 1)}, \psi_\lambda}, \ket{\widehat{\bot}_\lambda}\}$. Then we construct the circuit of Hamiltonian simulation by incorporating $\widehat{U}_H$ into the structure of QPP. Note that one eigenvalue of the Hamiltonian corresponds to two eigenvalues of $\widehat{U}_H$, and then the action of the circuit can be described as the matrix below, with respect to the basis $\{\ket{0,\chi_\lambda^{+}}, \ket{1,\chi_\lambda^{+}}\}$ and $\{\ket{0,\chi_\lambda^{-}}, \ket{1,\chi_\lambda^{-}}\}$ for each eigenvalue $\lambda$.
\begin{align}
    \bigoplus_{\lambda}\left(
    \begin{bmatrix}
        P(\tau_\lambda) & -Q(\tau_\lambda)\\
        Q^*(\tau_\lambda)& P^*(\tau_\lambda)
    \end{bmatrix}_\lambda
    \oplus
    \begin{bmatrix}
        P(-\tau_\lambda) & -Q(-\tau_\lambda)\\
        Q^*(-\tau_\lambda)& P^*(-\tau_\lambda)
    \end{bmatrix}_\lambda\right).
    \label{eqn:hs_circuit_matrix}
\end{align}
On the other hand, note that each state $\ket{0^{m+1}, \psi_\lambda}$ can be written as an equal-weighted sum of eigenvectors of $\widehat{U}_H$, implying the following equation.
\begin{align}
   \ket{0^{\otimes (m + 2)}, \psi_\lambda}= \frac{1}{\sqrt{2}}\left[\frac{1}{\sqrt{2}}\ket{0}(\ket{0^{\otimes (m + 1)},\psi_\lambda}+i \ket{\widehat{\bot}_\lambda})+\frac{1}{\sqrt{2}}\ket{0}(\ket{0^{\otimes (m + 1)},\psi_\lambda}-i \ket{\widehat{\bot}_\lambda})\right].
    \label{eqn:qubitization_rotation}
\end{align}
Using Eq.~\eqref{eqn:qubitization_rotation}, we thus rewrite the initial state as a superposition of eigenvectors of $\widehat{U}_H$. With the decomposition in Eq.~\eqref{eqn:hs_circuit_matrix}, applying QPP to the state $\ket{0^{m+2},\psi}$ outputs a state of the following form
\begin{align}
    \sum_{\lambda} \frac{\beta_\lambda}{\sqrt{2}} \left[
\frac{\left(P(\tau_\lambda)\ket{0}+Q^*(\tau_\lambda)\ket{1}\right)}{\sqrt{2}}(\ket{0^{\otimes (m + 1)},\psi_\lambda} + i \ket{\widehat{\bot}_\lambda}) + \frac{\left(P(-\tau_\lambda)\ket{0}+Q^*(-\tau_\lambda)\ket{1}\right)}{\sqrt{2}}(\ket{0^{\otimes (m + 1)},\psi_\lambda}-i \ket{\widehat{\bot}_\lambda})
    \right]
.\end{align}

The output state is near the target state as much as possible by suitably truncating the target function. In fact, the difference between the final output state and the target state is bounded by the quantity below.
\begin{align}
    \textrm{error}&\leq\max_{x\in[-\pi,\pi]}\left\|\left(P(x)\ket{0}+Q^*(x)\ket{1}\right)-e^{-i\cos(x)t}\ket{0}\right\|\\
    &=\max_{x\in[-\pi,\pi]}\sqrt{|P(x)-e^{-i\cos(x)t}|^2+1-|P(x)|^2}.
\end{align}
Let $a_x=P(x)-e^{-i\cos(x)t}$ denote the difference, and assume that $|a_x|\leq \eps$ for all $x\in[-\pi, \pi]$. Then we show how large $|P(x)|^2$ is.
\begin{align}
    |a_x|^2=(P(x)-e^{-i\cos(x)t})(P^*(x)-e^{i\cos(x)t})\Rightarrow |P(x)|^2&=2{\rm Re}(e^{i\cos(x)t}P(x))+|a_x|^2-1\\
    &\geq 1-2|a_x|+|a_x|^2=(1-|a_x|)^2.
\end{align}
The state approximation error is at most as large as
\begin{align}
    \textrm{error}\leq \max_{x\in[-\pi, \pi]}\sqrt{|a_x|^2+1-(1-|a_x|)^2}\leq\sqrt{2\eps}. 
    \label{eqn:error}
\end{align}
Hence, Eq.~\eqref{eqn:error} establishes the relation between the state approximation and the function approximation. 

The remaining is to show $P(x)$ can approximate $f(x)=e^{-i\cos(x)t}$ with arbitrary precision, which is true since $f(x)$ could be expanded into a trigonometric polynomial. We summarize the results in the following theorem and discuss how to truncate $f(x)$ in the proof.

\renewcommand\thetheorem{\ref{thm:ham simulation complexity}}
\setcounter{theorem}{\arabic{theorem}-1}
\begin{theorem}
Given a block encoding $U_H$ of $H/\Lambda$ for some $\Lambda \geq \norm{H}$, there exists an algorithm that simulates evolution under the Hamiltonian $H$ for time $t \in \mathbb{R}$ within precision $\delta>0$, using two ancilla qubits and querying controlled-$U_H$ and controlled-$U_H^\dagger$ for a total number of times in \[\Theta\left(\Lambda|t|+\frac{\log(2/\delta^2)}{\log\left(e+\frac{\log(2/\delta^2)}{\Lambda|t|}\right)}\right).\]
\end{theorem}
\begin{proof}
By Theorem~\ref{thm:qpp evolution}, QPP can simulate any trigonometric polynomial with an order $N$. Thus we just need to find such a polynomial that approximates the function $f(x)=e^{-i\cos(x)t}$ within the precision $\delta^2/2$. We recall the Jacobi-Anger expansion $e^{iz\cos(\theta)}=\sum_{k=-\infty}^{\infty}i^{k}J_{k}(z)e^{ik\theta}$ \cite{abramowitz1988handbook}, where $J_{k}(z)$ is the $k$-th Bessel function of the first kind. The truncation error of the Jacobi-Anger expansion has been well-studied in the literature \cite{Berry2015,gilyen2019quantum}. Given the truncation error $\delta^2/2$, the order $N$ is given by
\begin{align}
    N=\Theta\left(|z|+\frac{\log(2/\delta^2)}{\log\left(e+\frac{\log(2/\delta^2)}{|z|}\right)}\right).
\end{align}
Then QPP circuit uses $N$ times controlled $U_H$ and $U_H^\dagger$. Recall that the evolution time is $t\Lambda$, thus we set the parameter $z=t\Lambda$. Clearly, the cost of the circuit is the same as claimed, and the proof is finished.
\end{proof}
\renewcommand{\thetheorem}{S\arabic{theorem}}

\subsection{Hamiltonian eigensolver}
A fundamental problem in physics is to calculate static properties of a quantum
system. Of all the questions which one might ask about a quantum system, there is one most frequently asked: what are the energy and eigenstates of a Hamiltonian? After the development of quantum phase estimation, numerous quantum algorithms for computing the spectrum of a Hamiltonian have been developed. Clearly, we could just direct apply our phase search algorithm in a similar spirit to give quantum algorithms for extracting Hamiltonian eigen-information. To this end, we describe several quantum algorithms based on quantum phase estimation below.

The quantum algorithm proposed by~\citet{Abrams1999} computes the eigenstates and eigenvalues of a Hamiltonian via applying quantum phase estimation to the simulated real-time evolution operator, resulting in exponential speedups over known classical methods. Under post-selection of the ancilla register, this approach can output the spectrum by reading an eigenvalue and obtaining the corresponding eigenvector. 
Despite being conceptually simple, the technical realization of the synthesis between quantum phase estimation and Hamiltonian simulation could be difficult. Indeed, realizing the evolution operator requires a potentially large circuit, which would complicate the entire phase estimation circuit and make its practical implementation challenging.

In fact, the use of the time evolution operator is not necessary for phase estimation. Recent work by \citet{PhysRevLett.121.010501} proposed a method that extracts the spectrum of a Hamiltonian $H$ from a qubitized block encoding of $H/\L$. The idea is to estimate a phase of the qubitized block encoding $\widehat{U}$ then directly calculate the corresponding eigenvalue of $H$ by the relationship between their spectrum, i.e.\ $\lambda = \L \cos(\tau_\lambda)$.
To achieve the desired precision $\delta$, one needs to apply the block encoding for $\mathcal{O}(\Lambda/\delta)$ times.

Moreover, the eigenvector of the Hamiltonian can be extracted from the post-measurement state of quantum phase estimation. Specifically, suppose we obtain an eigenvector $\ket{\chi_\lambda^\pm}=\frac{1}{\sqrt{2}}(\ket{0^{\otimes (m + 1)}, \psi_\lambda}\pm i\ket{\bot_\lambda})$ by phase estimation. Since $(\ketbra{0^{\otimes (m + 1)}}{0^{\otimes (m + 1)}} \otimes I^{\otimes n}) \ket{\bot_\lambda}=0$, we directly measure the ancilla register of the qubitized block encoding in the computational basis. Clearly, the probability of receiving all zeros from the measurement is exactly half, and the post-measurement state is the corresponding eigenvector. If the result is not as expected, we could re-apply the phase estimation and repeat the measurement, projecting the state into $\ket{\chi_\lambda^+}$ or $\ket{\chi_\lambda^-}$. Repeating for $k$ times makes the failure probability decay exponentially as $(1/2)^{k}$.

Another important Hamiltonian problem is to compute ground and excited state properties of many-body systems, which is quite a challenge in quantum physics and quantum chemistry. Recently, \citet{dong2022ground} proposed a method via quantum eigenvalue transformation of unitary matrices, which could be naturally generalized into the structure of QPP.

\section{Theorems of Entropy Estimation} \label{appendix:entropy estimation}
\subsection{Proof of Theorem~\ref{thm:state entropy estimation}} \label{appendix:state entropy estimation}
\renewcommand\thetheorem{\ref{thm:state entropy estimation}}
\setcounter{theorem}{\arabic{theorem}-1}
\begin{theorem}
    Let $\ket{\Psi_\rho}_{AB}$ be a purification of an $n$-qubit state $\rho$ and $\widehat{U}_\sigma$ be a qubitized block encoding of an $n$-qubit state $\sigma$ with $m$ ancilla qubits. For any real-valued polynomial $f(x) = \sum_{k=0}^L c_j x^k$ with \highlight{$\abs{F(x)} \leq 1$ for all $x \in \RR$}, there exists a QPP circuit $\UQPPs(\widehat{U}_\sigma)$ of $L$ layers such that 
    \begin{align}
        \langle Z^{(0)} \rangle_{\ket{\Phi}} = \tr \left(\rho f(\sigma) \right)
    ,\end{align}
    where $\ket{\Phi} = \left(\UQPPs(\widehat{U}_\sigma) \otimes I_B\right)\ket{0^{\otimes (m + 1)}} \ket{\Psi_\rho}_{AB}$ and the polynomial on a quantum state is defined as $f(\sigma) = \sum_{k=0}^L c_j \sigma^k$.
\end{theorem}
\begin{proof}
We start with the spectral decomposition of $\sigma$ and $\widehat{U}_\sigma$ under subspace $\cH_\sigma$ defined in Eq.~\eqref{eqn:hilbert decomposition}. From equation Eq.~\eqref{eqn:spectral} we have
\begin{align}
    \sigma &= \sum_{j=0}^{2^n - 1} q_j \ketbra{\psi_j}{\psi_j},\\
    \quad \widehat{U}_\sigma &= \bigoplus_{j=0}^{2^n - 1}
      \begin{bmatrix}
          e^{i \tau_j} & 0\\
          0 & e^{-i \tau_j}
      \end{bmatrix}_{\{\ket{\chi^+_j}, \ket{\chi^-_j}\}} \oplus [\cdots]_{\cH_\sigma^\bot}
,\end{align}
where $\tau_j = \arccos(q_j)$, and $\ket{\chi^\pm_j}$ are eigenstates of $\widehat{U}_\sigma$ such that 
\begin{equation} \label{eqn:superpositioin of sigma}
\begin{matrix}
    \ket{\chi_j^{+}} = \frac{1}{\sqrt{2}}(\ket{0^{\otimes m}, \psi_j} + i\ket{\widehat{\bot}_j});\\
    \ket{\chi_j^{-}} = \frac{1}{\sqrt{2}}(\ket{0^{\otimes m}, \psi_j} - i\ket{\widehat{\bot}_j})
\end{matrix}
\end{equation}
for some quantum state $\ket{\widehat{\bot}_j}$ (defined in Appendix~\ref{appendix:qubitization intro}) so that $(\ketbra{0^{\otimes m}}{0^{\otimes m}} \otimes I^{\otimes n}) \ket{\widehat{\bot}_j}  = 0$. 
Also, note that $\ketbra{0^{\otimes m}}{0^{\otimes m}} \otimes \rho$ is a density matrix in $\cH_\sigma$ and hence can be decomposed by the basis $\{ \ket{\chi^\pm_j} \}_{j=0}^{2^n - 1}$. Now we are ready to analyze the effect of QPP on the input state. Define the function $F(x) = \sum_{k=-L}^L d_k e^{ikx} \coloneqq f(\cos(x))$. Note that $\abs{F}(x) \leq 1$ implies $\norm{\bm d}_1 \leq 1$.

Suppose the input state is a purified state $\ket{\Psi_\rho}_{AB}$ so that $\tr_B(\ket{\Psi_\rho}_{AB}) = \rho$.
Here, the register $A$ is the main register of the QPP circuit. By Schmidt decomposition, there exists an orthonormal set $\{\ket{\phi_j}\}_{j=0}^{2^n - 1}$ of quantum states on register $B$ such that
\begin{align}
   \ket{\Psi_\rho}_{AB} = \sum_{j=0}^{2^n - 1} \sqrt{p_j} \ket{\psi_j}_A \ket{\phi_j}_B
.\end{align}
From Theorem~\ref{thm:qpp evaluation}, there exists a QPP circuit $\UQPPs(\widehat{U}_\sigma \otimes I_B)$ of $L$ layers such that
\begin{align}
        \langle Z^{(0)}  \rangle_{\ket{\Phi}} &= \tr\left[\left(Z^{(0)} \otimes I^{\otimes (m + n)} \otimes I_B\right) \cdot \ket{\Phi} \bra{\Phi} \right] \\
        & = \sum_{j=0}^{2^n - 1} p^+_j F(\tau_j) + \sum_{j=0}^{2^n - 1} p^-_j F(-\tau_j) = \sum_{j=0}^{2^n - 1} (p^+_j + p^-_j) f(q_j)
,\end{align}
where $p^\pm_j = \bra{\chi^\pm_j}\, (\ketbra{0^{\otimes m}}{0^{\otimes m}} \otimes \rho) \,\ket{\chi^\pm_j}$. Further, Eq.~\eqref{eqn:superpositioin of sigma} implies
\begin{align}
    p^+_j + p^-_j & = \bra{\chi^+_j} \left(\ketbra{0^{\otimes m}}{0^{\otimes m}} \otimes \rho\right) \ket{\chi^+_j} + \bra{\chi^-_j} \left(\ketbra{0^{\otimes m}}{0^{\otimes m}} \otimes \rho \right) \ket{\chi^-_j} \\
    & = \bra{\psi_j} \rho \ket{\psi_j}
.\end{align}
The statement holds as
\begin{equation}
    \langle Z^{(0)}  \rangle_{\ket{\Phi}} = \sum_{j=0}^{2^n - 1} \bra{\psi_j} \rho \ket{\psi_j} f(q_j) = \tr \left(\rho f(\sigma) \right)
.\end{equation} 
\end{proof}\\
\renewcommand{\thetheorem}{S\arabic{theorem}}

That is, for any entropic function that is well approximated by a polynomial, we can use QPP circuits to estimate such an entropy by Theorem~\ref{thm:state entropy estimation}. The following result guarantees that the number of circuit layers $L$ for polynomial approximation is $\cO \left(\log \frac{1}{\eps}\right)$ up to precision $\eps$.

\begin{lemma} [Corollary 66 in~\cite{gilyen2019quantum}] \label{lem:gurantee for poly approx}
Let $x_0 \in [-1, 1]$, $r \in (0, 2]$, $\delta \in (0, r]$, and let $f: [-x_0 - r - \delta, x_0 + r + \delta] \to \CC$ and be such that $f(x_0 + x) = \sum_{k=0}^{\infty} c_k x^k$ for all $x \in [-r - \delta, r + \delta]$. Suppose $B > 0$ is such that $\sum_{k=0}^{\infty} (r + \delta)^k |c_k| \leq B$. Let $\eps \in (0, \frac{1}{2B}]$, then there is an efficiently computable polynomial $P \in \CC[x]$ of degree $\cO(\frac{1}{\delta} \log \frac{B}{\eps})$ such that
\begin{itemize}
    \item for all $x \in [-1, 1]$, $\abs{P(x)} \leq \eps + B$ and
    \item for all $x \in [x_0 - r, x_0 + r]$, $\abs{f(x) - P(x)} \leq \eps$.
\end{itemize}
\end{lemma}

\subsection{Proofs for von Neumann and relative entropy estimation} \label{appendix:guarantee for neumann and relative}

\renewcommand\thetheorem{\ref{thm:von Neumann complexity}}
\setcounter{theorem}{\arabic{theorem}-1}
\begin{theorem}[von Neumann entropy estimation]
    Given a purified quantum query oracle $U_\rho$ of a state $\rho$ whose non-zero eigenvalues are lower bounded by $\gamma > 0$, there exists an algorithm that estimates $S(\rho)$ up to precision $\eps$ with high probability by measuring a single qubit, querying $U_\rho$ and $U_\rho^\dagger$ for $\cO\left(\frac{1}{\gamma \eps^2}\log^2(\frac{1}{\gamma})\log(\frac{\log(1/\gamma)}{\eps})\right)$ times. Moreover, using amplitude estimation improves the query complexity to $\cO\left(\frac{1}{\gamma \eps}\log(\frac{1}{\gamma})\log(\frac{\log(1/\gamma)}{\eps})\right)$.
\end{theorem}
\begin{proof}
Denote $f(x) = \frac{\ln(x)}{2\ln(\gamma)}$. We expect to find a polynomial $P(x) \in \RR[x]$ such that for all $x \in [\gamma, 1]$,
\begin{equation}
    |P(x) - f(x)| \leq \frac{\eps}{4 \ln(1 / \gamma)}
.\end{equation}
By taking $x_0=1$, $r=1-\gamma$, $\delta=\frac{\gamma}{2}$, and $B = \frac{1}{2}$ and $\eta = \frac{\eps}{4 \ln(1/\gamma)}$ into Lemma~\ref{lem:gurantee for poly approx}, such polynomial $P$ exists with degree 
\begin{equation}
    L = \cO\left(\frac{1}{\delta} \log \frac{B}{\eta}\right) 
    = \cO\left(\frac{1}{\gamma} \log \frac{\log(1/\gamma)}{\eps}\right) 
.\end{equation} 
Note that Lemma 11 of previous work~\cite{gilyen2019quantum} used the same setup for Lemma~\ref{lem:gurantee for poly approx}. Then Theorem~\ref{thm:state entropy estimation} implies that there exists a QPP circuit $\UQPPs(\widehat{U}_\rho)$ of $L$ layers to estimate $\tr(\rho P(\rho))$. Up to precision $\frac{\eps}{4 \ln(1 / \gamma)}$, the experimental estimation $E$ can be retrieved by measuring the first qubit of $\UQPP(\widehat{U}_\rho)$ with input state $\ketbra{0^{\otimes (m + 1)}}{0^{\otimes (m + 1)}} \otimes \rho$. We have
\begin{align}
    |E - \tr(\rho f(\rho))| & \leq |E - \tr(\rho P(\rho))| + |\tr(\rho P(\rho)) - \tr(\rho f(\rho))| \\
    & \leq \frac{\eps}{4 \ln(1/\gamma)} + \norm{P - f}_{[\gamma, 1]} \leq \frac{\eps}{2 \ln(1/\gamma)}
\end{align}
and hence
\begin{equation}
    |2 \ln(\gamma) E - S(\rho)| \leq \eps
.\end{equation}
To receive the desired precision, by Chebyshev's inequality, the total number of measurements is $\cO\left(\frac{\ln^2(\gamma)}{\eps^2}\right)$, while each $\widehat{U}_\rho$ requires one call of $U_\rho$ and its inverse. Consequently, the total query complexity of $U_\rho$ and its inverse is 
\begin{equation}
    \cO\left(\frac{1}{\gamma \eps^2}\log^2(\frac{1}{\gamma})\log\frac{\log(1/\gamma)}{\eps}\right)
.\end{equation}
The statement for using amplitude estimation follows by switching Chebyshev's inequality to the complexity of amplitude estimation~\cite{brassard2002quantum, grinko2021iterative}.
\end{proof}\\
\renewcommand{\thetheorem}{S\arabic{theorem}}

Note that by replacing $\widehat{U}_\rho$ with $\widehat{U}_\sigma$, the proof of Theorem~\ref{thm:relative complexity} is exactly the same as the proof of Theorem~\ref{thm:von Neumann complexity}.

\renewcommand\thecorollary{\ref{coro:von Neumann complexity no gamma}}
\setcounter{corollary}{\arabic{corollary}-1}
\begin{corollary}
    Given a purified quantum query oracle $U_\rho$ of a state $\rho$ whose rank is $\kappa \leq 2^n$, there exists an algorithm that estimates $S(\rho)$ up to precision $\eps$ with high probability by measuring a single qubit, querying $U_\rho$ and $U_\rho^\dagger$ for $\cO\left( \frac{\kappa}{\eps^3} \log^3 \left(\frac{\kappa}{\eps}\right)\log\left(\frac{1}{\eps}\right)\right)$ times. Moreover, using amplitude estimation improves the query complexity to $\cO\left( \frac{\kappa}{\eps^2} \log^2 \left(\frac{\kappa}{\eps}\right) \log\left(\frac{1}{\eps}\right)\right)$.
\end{corollary}

\begin{proof}
We follow the same proof in Theorem~\ref{thm:von Neumann complexity}, with extra consideration on the threshold value $\gamma > 0$ and the error generated by this threshold value. We choose the same function $f$ and a polynomial $P$ of degree $\cO(\frac{1}{\gamma} \log \frac{1}{\eta})$ such that for all $x \in [\gamma, 1]$,
\begin{equation}
    \abs{P(x) - f(x)} \leq \eta
,\end{equation}
where $\eta > 0$ is decided later. Denote $E$ as the experiment estimation. Then
\begin{align}
    \abs{E - \tr(\rho f(\rho))} &\leq \abs{E - \tr(\rho P(\rho))} + \abs{\tr(\rho P(\rho)) - \tr(\rho f(\rho))} \\
    &\leq \abs{E - \tr(\rho P(\rho))} + \kappa \norm{x P - x f}_{[0, \gamma]} + \norm{P - f}_{[\gamma, 1]} \\
    &\leq \abs{E - \tr(\rho P(\rho))} + 2\kappa\gamma+ \eta
.\end{align}
Choose $\abs{E - \tr(\rho P(\rho))} \leq \frac{\eps}{8\ln(1/\gamma)}$, $\gamma = \frac{\eps}{16\kappa\ln(\kappa/\eps)}$ and $\eta = \frac{\eps}{8\ln(1/\gamma)}$. We have
\begin{equation}
    \abs{2\ln(\gamma)E -S(\rho)} \leq \frac{\eps}{4} + \frac{\eps}{2} + \frac{\eps}{4} = \eps
.\end{equation}
Then the total complexity is 
\begin{equation}
      \cO\left(\frac{1}{\eps^2} \log^2 \left(\frac{\kappa}{\eps}\right)\right) \cdot \cO\left(\frac{\kappa}{\eps} \log \left(\frac{\kappa}{\eps}\right)\log\left(\frac{1}{\eps}\right)\right) = \cO\left( \frac{\kappa }{\eps^3} \log^3\left(\frac{\kappa}{\eps}\right) \log\left(\frac{1}{\eps}\right) \right)
,\end{equation}
as required. The statement for using amplitude estimation follows by switching Chebyshev's inequality to the complexity of amplitude estimation~\cite{brassard2002quantum, grinko2021iterative}.
\end{proof}
\renewcommand{\thecorollary}{S\arabic{corollary}}

\subsection{Proofs for quantum R\'enyi entropy estimation}\label{appendix:guarantee for Renyi}

An extra lemma is required for proceeding to the main content. Note that this lemma is the result of Lemma~\ref{lem:gurantee for poly approx}.

\begin{lemma} [Corollary 67 in \cite{gilyen2019quantum}] \label{lem:gurantee for neg power approx}
Suppose $\gamma, \eps \in (0, \frac{1}{2})$ and $c \in (0, 1)$. Then there exists a polynomial $P \in \RR[x]$ of degree $\cO(\frac{1}{\gamma} \log \frac{1}{\eps})$ such that
\begin{itemize}
    \item for all $ x \in [-1, 1]$, $\abs{P(x)} \leq 1$ and
    \item for all $x\in [\gamma, 1]$, $\abs{P(x) - \frac{\gamma^c}{2} x^{-c}} \leq \eps$.
\end{itemize}
\end{lemma}

\renewcommand\thetheorem{\ref{thm:real Renyi complexity}}
\setcounter{theorem}{\arabic{theorem}-1}
\begin{theorem} [Quantum R\'enyi entropy estimation for real $\alpha$]
 Given a purified quantum query oracle $U_\rho$ of a state $\rho$ whose non-zero eigenvalues are lower bounded by $\gamma > 0$, there exists an algorithm that estimates $S_\alpha(\rho)$ up to precision $\eps$ with high probability by measuring a single qubit, querying $U_\rho$ and $U_\rho^\dagger$ for a total number of times in
\begin{equation}
    \begin{dcases}
        \cO \left(\frac{1}{\gamma^{3 - 2\a} \eps^2} \log\left(\frac{\tr(\rho^\a)}{\gamma^{1 - \a}\eps}\right) \cdot \eta^2 \right), & \textrm{if } \a \in (0, 1); \\[10pt]
        \cO \left(\frac{\log^2(1/\gamma)}{\gamma \eps^2} \left[\a \gamma + \log \left(\frac{\tr(\rho^\a) \log(1/\gamma)}{\eps}\right)\right] \cdot \eta^2\right), & \textrm{if } \a \in (1, \infty);
    \end{dcases}
\end{equation}
where $\eta = \frac{\tr(\rho^\a)^{-1}}{\abs{1 - \a}}$.
Moreover, using amplitude estimation improves the query complexity to
\begin{equation}
    \begin{dcases}
        \cO \left(\frac{1}{\gamma^{2 - \a} \eps} \log\left(\frac{\tr(\rho^\a)}{\gamma^{1 - \a}\eps}\right) \cdot \eta \right), & \textrm{if } \a \in (0, 1). \\[10pt]
        \cO \left(\frac{\log(1/\gamma)}{\gamma \eps} \left[\a \gamma + \log \left(\frac{\tr(\rho^\a) \log(1/\gamma)}{\eps}\right)\right]  \cdot \eta \right), & \textrm{if } \a \in (1, \infty).
    \end{dcases}
\end{equation}
\end{theorem}
\begin{proof}
Previous work~\cite{wang2022quantum} states that it is able to obtain an estimation of $S_\a(\rho)$ up to precision $\eps$, by an estimation of $\tr(\rho^\a)$ within error $\eps' = \frac{\abs{1 - \alpha} \tr(\rho^\a)}{2} \eps$. Then we turn to demonstrate how to obtain $\tr(\rho^\a)$ with error bounded by $\eps'$.

Suppose $\a \in (0, 1)$. By Lemma~\ref{lem:gurantee for neg power approx}, there exists a polynomial $P \in \RR[x]$ of degree $\cO(\frac{1}{\gamma} \log \frac{1}{\gamma^{1 - \a} \eps'})$ such that
\begin{itemize}
    \item for all $x\in [-1, 1]$, $\abs{P(x)} \leq 1$,
    \item for all $x \in [\gamma, 1]$, $\abs{P(x) - \frac{\gamma^{1 - \a}}{2} x^{\a - 1}} \leq \frac{\gamma^{1 - \a}}{4}\eps'$.
\end{itemize}
Similar to the proof of Theorem~\ref{thm:von Neumann complexity}, Theorem~\ref{thm:state entropy estimation} implies that there exists a QPP circuit $\UQPPs(\widehat{U}_\rho)$ of $L$ layers to estimate $\tr(\rho P(\rho))$. Up to precision $\frac{\gamma^{1 - \a}}{4} \eps'$, the experimental estimation $E$ can be retrieved by measuring the first qubit of $\UQPPs(\widehat{U}_\rho)$ with input state $\ketbra{0^{\otimes (m + 1)}}{0^{\otimes (m + 1)}} \otimes \rho$. We have
\begin{align}
    |E - \frac{\gamma^{1 - \a}}{2} \tr(\rho^\a)| \leq \frac{\gamma^{1 - \a}}{2}\eps'
\end{align}
and hence
\begin{equation}
    |2\gamma^{\a - 1} E - \tr(\rho^\a)| \leq \eps'
.\end{equation}
By considering Chebyshev's inequality, the total query complexity of $U_\rho$ and $U_\rho^\dagger$ is $\cO (\frac{1}{\gamma^{3 - 2\a} \eps'^2}\log\frac{1}{\gamma^{1 - \a} \eps'})$. Now substitute $\eps'$ back to $\eps$. The query complexity turns to
\begin{equation}
    \cO \left(\frac{\tr(\rho^\a)^{-2}}{\abs{1 - \a}^2 \gamma^{3 - 2\a} \eps^2} \log\left(\frac{\tr(\rho^\a)}{\gamma^{1 - \a}\eps}\right)\right)
.\end{equation}

Now suppose $\a > 1$. Define $f(x) \coloneqq \frac{1}{2\ln(2e/\gamma)} 
x^{\a - \lfloor \a \rfloor}$ which is bounded by $\frac{1}{2\ln(2e/\gamma)}$ for $x \in [0, 1]$. Using the same setup in the proof of Theorem~\ref{thm:von Neumann complexity}, we would have $B= \frac{1}{2}$. See deductions below.
\begin{align}
    \sum_{l=0}^{\infty} \left|\binom{\alpha-\lfloor \alpha\rfloor}{l}\right| (1-\gamma/2)^l & \leq \sum_{l=0}^{\infty} \frac{\alpha-\lfloor \alpha\rfloor}{l}(-1+\gamma/2)^l \\
    &=1-(\alpha-\lfloor \alpha\rfloor)\sum_{l=1}^{\infty} \frac{(-1)^{l-1}}{l}(-1+\gamma/2)^l \\
    &=1-(\alpha-\lfloor \alpha\rfloor) \ln(\gamma/2)\leq \ln(2e/\gamma).
\end{align}
Immediately, Lemma~\ref{lem:gurantee for poly approx} implies that there exists a polynomial $\widetilde{P}(x) \in \RR[x]$ of degree $\cO(\frac{1}{\gamma} \log \frac{\log(1/\gamma)}{\eps})$ such that for all $x \in [\gamma, 1]$,
\begin{equation}
    |\widetilde{P}(x) - f(x)| \leq \frac{\eps'}{4\ln(2e/\gamma)}
.\end{equation}
Redefine $P(x) \coloneqq x^{\lfloor \a \rfloor} \widetilde{P}(x)$ of degree $\cO(\alpha + \frac{1}{\gamma} \log \frac{\log(1/\gamma)}{\eps})$. By the same procedure in the case of $\a \in (0, 1)$, we can find a QPP circuit estimating $\tr(\rho P(\rho))$, so that the experimental estimation $E$ satisfies
\begin{equation}
    |2\ln(2e/\gamma)E - \tr(\rho^\a)| \leq \eps'
.\end{equation}
As a result, the total query complexity in terms of $\eps$ is
\begin{equation}
    \cO \left(\frac{\tr(\rho^\a)^{-2}\log^2(1/\gamma)}{\abs{1 - \a}^2 \gamma \eps^2} \left(\a \gamma + \log\frac{\tr(\rho^\a) \log(1/\gamma)}{\eps}\right)\right)
.\end{equation}
The statement for using amplitude estimation follows by switching Chebyshev's inequality to the complexity of amplitude estimation.
\end{proof}\\
\renewcommand{\thetheorem}{S\arabic{theorem}}

As for the proof of Theorem~\ref{thm:integer Renyi complexity}, Theorem~\ref{thm:state entropy estimation} implies the query complexity is simply $\cO(\frac{\alpha}{\eps'^2})$ when $\a$ is an integer. Then the statement of Theorem~\ref{thm:integer Renyi complexity} follows by replacing $\eps'$ with $\eps$. Note that in this case $\abs{1 - \a} \in \cO(\a)$. In the following corollary, we present a method to estimate $S_\a(\rho)$ when $\gamma$ is unknown.

\begin{corollary} \label{coro:renyi complexity no gamma}
Assume a rank $\kappa \leq 2^{n}$ for an $n$-qubit state $\rho$ and a purified quantum oracle $U_\rho$. There exists a QPP circuit that estimates $S_\alpha(\rho)$ within precision $\eps$ by measuring a single qubit, at an expense of querying $U_\rho$ and $U_\rho^\dagger$ for a number of 
\begin{equation}
    \begin{dcases}
        \cO\left( \frac{2^{12/\a} \kappa^{3/\a - 2}}{\a \eps^{3/\a}} \left[3 + \log\left( \frac{\tr(\rho^\a) \kappa}{\abs{1 - \a} \eps} \right)\right] \cdot \eta^{3/\a}\right), & \textrm{if } \a \in (0, 1), \\[10pt]
        \cO\left(  \frac{\kappa}{\eps^3} \log^2 \left(\frac{\tr(\rho^\a)^{-1} \kappa}{\abs{1 - \a} \eps}\right) \left[ \frac{\a^2 \eps}{2\kappa} + \log \left(\frac{\tr(\rho^\a)^{-1} \kappa}{\abs{1 - \a} \eps}\right)\log\left(\frac{\tr(\rho^\a)^{-1}}{\abs{1 - \a} \eps}\right) \right] \cdot \eta^3\right), & \textrm{if } \a \in (1, \infty),
    \end{dcases}
.\end{equation}
where $\eta = \frac{\tr(\rho^\a)^{-1}}{\abs{1 - \a}}$. 
Moreover, using amplitude estimation improves query complexity to
\begin{equation}
    \begin{dcases}
        \cO\left( \frac{2^{8/\a} \kappa^{2/\a - 1}}{\a \eps^{2/\a}} \left[3 + \log\left( \frac{\tr(\rho^\a) \kappa}{\abs{1 - \a} \eps} \right)\right] \cdot \eta^{2/a} \right), & \textrm{if } \a \in (0, 1). \\[10pt]
        \cO\left(  \frac{\kappa}{\eps^2} \log \left(\frac{\tr(\rho^\a)^{-1} \kappa}{\abs{1 - \a} \eps}\right) \left[ \frac{\a^2 \eps}{2\kappa} + \log \left(\frac{\tr(\rho^\a)^{-1} \kappa}{\abs{1 - \a} \eps}\right)\log\left(\frac{\tr(\rho^\a)^{-1}}{\abs{1 - \a} \eps}\right) \right]  \cdot \eta^{2} \right), & \textrm{if } \a \in (1, \infty).
    \end{dcases}
\end{equation}
\end{corollary}
\begin{proof}
As shown in the proof of Theorem \ref{thm:real Renyi complexity}, we could find a polynomial $P(x)$ such that 
\begin{equation}
    \abs{P(x) - f_\alpha(x)} \leq \eta, \forall x\in[\gamma ,1].
\end{equation}
where 
\begin{align}
    f_\alpha(x)=
    \begin{dcases}
   \frac{\gamma^{1 - \a}}{2} x^{\a - 1}, & \text{if } \alpha\in(0,1), \\
    \frac{x^\alpha}{2\ln(2e/\gamma)}, & \text{if } \alpha>1.
    \end{dcases}
\end{align}
Denote $E$ as the experiment estimation. Then
\begin{align}
    \abs{E - \tr(\rho f(\rho))} &\leq \abs{E - \tr(\rho P(\rho))} + \abs{\tr(\rho P(\rho)) - \tr(\rho f(\rho))} \\
    &\leq \abs{E - \tr(\rho P(\rho))} + \kappa \norm{x P - x f}_{[0, \gamma]} + \norm{P - f}_{[\gamma, 1]} \\
    &\leq \abs{E - \tr(\rho P(\rho))} + 2\kappa\gamma+ \eta
.\end{align}

For $\alpha>1$, choose $\abs{E - \tr(\rho P(\rho))} \leq \frac{\eps'}{16\ln(2e/\gamma)}$, $\gamma = \frac{\eps'}{16\kappa\ln(16\kappa/\eps')}$ and $\eta =  \frac{\eps'}{16\ln(2e/\gamma)}$. We have
\begin{equation}
    \abs{2\ln(2e/\gamma)E -\tr(\rho^\alpha)} \leq \frac{\eps'}{4} + \frac{\eps'}{2} + \frac{\eps'}{4} = \eps'
.\end{equation}
Then the total complexity is 
\begin{align}
      & \quad \,\, \cO\left(\frac{1}{(\eps')^2} \log^2 \left(\frac{\kappa}{\eps'}\right)\right) \cdot \cO\left(\alpha+\frac{\kappa}{\eps'} \log \left(\frac{\kappa}{\eps'}\right)\log\left(\frac{1}{\eps'}\right)\right) \\
      & = \cO\left(  \frac{\kappa}{(\eps')^3} \log^2 \left(\frac{\kappa}{\eps'}\right) \left[ \frac{\a \eps'}{\kappa} + \log \left(\frac{\kappa}{\eps'}\right)\log\left(\frac{1}{\eps'}\right) \right] \right) \\
      & = \cO\left(  \frac{\tr(\rho^\a)^{-3} \kappa}{\abs{1 - \a}^3 \eps^3} \log^2 \left(\frac{\tr(\rho^\a)^{-1} \kappa}{\abs{1 - \a} \eps}\right) \left[ \frac{\a \abs{1 - \a} \tr(\rho^\a) \eps}{2\kappa} + \log \left(\frac{\tr(\rho^\a)^{-1} \kappa}{\abs{1 - \a} \eps}\right)\log\left(\frac{\tr(\rho^\a)^{-1}}{\abs{1 - \a} \eps}\right) \right] \right)
.\end{align}
The result follows by the fact that $\a \abs{1- \a} = \cO(\a^2)$ and $\tr(\rho^\a) \leq 1$.

For $\alpha\in(0,1)$, choose $\abs{E - \tr(\rho P(\rho))} \leq \frac{\gamma^{1-\a}\eps'}{8}$, $\gamma = 8^{-1/\a} \cdot \left(\frac{\eps'}{\kappa}\right)^{1/\a}$ and $\eta =\frac{\gamma^{1-\a}\eps'}{8}$. We have
\begin{equation}
    \abs{2\gamma^{\alpha-1}E -\tr(\rho^\alpha)} \leq \frac{\eps'}{4} + \frac{\eps'}{2} + \frac{\eps'}{4} = \eps'
.\end{equation}
Then the total complexity is 
\begin{align}
    & \quad \,\, \cO\left(\frac{\gamma^{2(\alpha-1)}}{(\eps')^2}\right) \cdot \cO\left( \frac{1}{\gamma} \log \left(\frac{\gamma^{\alpha-1}}{\eps'}\right)\right) \\
    & = \cO\left( \frac{2^{6/\a} \kappa^{2/\a - 2}}{(\eps')^{2/\a}} \right) \cdot \cO\left(\frac{2^{3/\a} \kappa^{1/\a}}{(\eps')^{1/\a}} \log \left(\frac{8^{1/\a-1} \kappa^{1/\a - 1}}{(\eps')^{1/\a}}\right)\right) = \cO\left( \frac{2^{9/\a} \kappa^{3/\a - 2}}{\a(\eps')^{3/\a}} \left[3 + \log\left( \frac{\kappa}{\eps'} \right)\right]\right) \\
    & = \cO\left( \frac{2^{12/\a} \tr(\rho^\a)^{-3/\a} \kappa^{3/\a - 2}}{\a\abs{1 - \a}^{3/\a} \eps^{3/\a}} \left[3 + \log\left( \frac{\tr(\rho^\a) \kappa}{\abs{1 - \a} \eps} \right)\right]\right)
,\end{align}
as required. 
The statement for using amplitude estimation follows by switching Chebyshev's inequality to the complexity of amplitude estimation.
\end{proof} \\

\subsection{Comparison on entropies estimation algorithms} \label{appendix:entropy table}

We present further comparison on different entropies estimation algorithms in this section. In Table~\ref{table:complete von Neumann algs}, we add the von Neumann entropy estimation algorithm proposed by~\citet{wang2022new}. When assuming rank $\kappa$, the QPP-based algorithm improves the result in~\cite{wang2022new} by a factor of $\kappa$. In Table~\ref{table:complete Renyi algs}, we add R\'enyi entropy estimation algorithms proposed by~\citet{wang2022new}.

\begin{table}[htbp]
\centering
\setlength{\tabcolsep}{1em}
\begin{tabular}{lccc}
\toprule
Methods for $S(\rho)$ estimation & Total queries to $U_\rho$ and $U_\rho^\dagger$ & Queries per use of circuit\\
\midrule
QSVT-based with QAE (\cite{gilyen2020distributional}) & $\widetilde\cO(\frac{d}{\eps^{1.5}})$ & $\widetilde\cO(\frac{d}{\eps^{1.5}})$ \\
\addlinespace
QSVT-based with QAE (assumes rank, \cite{wang2022new}) & $\widetilde\cO(\frac{\kappa^2}{\eps^2})$ & $\widetilde\cO(\frac{\kappa^2}{\eps^2})$   \\
\addlinespace
QPP-based (assumes rank, in Corollary~\ref{coro:von Neumann complexity no gamma}) & $\widetilde\cO(\frac{\kappa}{\eps^3})$ & $\widetilde\cO(\frac{\kappa}{\eps})$  \\
\addlinespace
QPP-based with QAE (assumes rank, in Corollary~\ref{coro:von Neumann complexity no gamma})  & $\widetilde\cO(\frac{\kappa}{\eps^2})$ &  $\widetilde\cO(\frac{\kappa}{\eps^2})$  \\
\addlinespace
QPP-based (in Theorem~\ref{thm:von Neumann complexity})  & $\widetilde\cO(\frac{1}{\gamma \eps^2})$ & $\widetilde\cO(\frac{1}{\gamma})$  \\
\addlinespace
QPP-based with QAE (in Theorem~\ref{thm:von Neumann complexity}) & $\widetilde\cO(\frac{1}{\gamma \eps})$ &  $\widetilde\cO(\frac{1}{\gamma\eps})$  \\
\bottomrule
\end{tabular}
\caption{Comparison of algorithms on estimating von Neumann entropy within additive error. Here the $\widetilde\cO$ notation omits $\log$ factors, $\gamma > 0$ is the lower bound of eigenvalues, $\kappa > 0$ is the rank of the state $\rho \in \CC^{d \times d}$, and $\eps$ is the additive error of estimating $S(\rho)$. QAE is short for quantum amplitude estimation.} 
\label{table:complete von Neumann algs}
\end{table}

\begin{table}[H]
\setlength{\tabcolsep}{1em}
\centering
\begin{tabular}{lccc}
\toprule
\multirow{2}{*}[-2pt]{Methods for $S_\a(\rho)$ estimation} & \multicolumn{3}{c}{Total queries to $U_\rho$ and $U_\rho^\dagger$}  \\
\cmidrule{2-4}
{} & $\a \in (0, 1)$ & $\a \in (1, \infty)$ &  $\a \in \NN_+$\\
\midrule
QSVT-based with $\mathsf{DQC1}$ (\cite{subramanian2021quantum}) & $\widehat{\cO} \left(\frac{d^2}{\gamma \eps^2} \cdot \eta^2\right)$ & $\widehat{\cO} \left(\frac{d^2}{\gamma \eps^2} \cdot \eta^2\right)$ & $\widehat{\cO} \left(\frac{d^2}{\eps^2} \cdot \eta^2\right)$ \\
\addlinespace
QSVT-based with QAE (assumes rank, \cite{wang2022new}) & $\widetilde{\cO} \left(\frac{\kappa^{\frac{3 - \a^2}{2\a}}}{\eps^{\frac{3 + \a}{2\a}}}\right)$ & $\widetilde{\cO} \left(\frac{\kappa^{\a - 1 + \a/\{ \frac{\a - 1}{2} \}}}{\eps^{1 + 1/\{ \frac{\a - 1}{2} \}}}\right)$ & $\widetilde{\cO} \left(\frac{\kappa^{\a - 1}}{\eps}\right)$ (only for odd $\a$)\\
\addlinespace
QPP-based (assumes rank, in Corollary~\ref{coro:renyi complexity no gamma}) & $\widehat{\cO}\left( \frac{\kappa^{3/\a - 2}}{\eps^{3/\a}} \cdot \eta^{3/\a}\right)$ & $\widehat{\cO}\left(  \frac{\alpha^2\eps + \kappa}{\eps^3} \cdot \eta^{3} \right)$ & $\widehat{\cO} \left(\frac{1}{\eps^2} \cdot \eta^2 \right)$ \\
\addlinespace
QPP-based with QAE (assumes rank, in Corollary~\ref{coro:renyi complexity no gamma}) & $\widehat{\cO}\left( \frac{\kappa^{2/\a - 1}}{\eps^{2/\a}} \cdot \eta^{2/\a}\right)$ & $\widehat{\cO}\left(  \frac{\alpha^2\eps + \kappa}{\eps^2} \cdot \eta^2 \right)$ & $\widehat{\cO} \left(\frac{1}{\eps} \cdot \eta \right) $ \\
\addlinespace
QPP-based (in Theorem~\ref{thm:real Renyi complexity} and~\ref{thm:integer Renyi complexity}) & $\widehat{\cO} \left(\frac{1}{\gamma^{3 - 2\a} \eps^2}\cdot \eta^2\right)$ & $\widehat{\cO} \left(\frac{\a\gamma + 1}{\gamma \eps^2}\cdot \eta^2\right)$ & $\widehat{\cO} \left(\frac{1}{\eps^2} \cdot \eta^2\right)$\\
\addlinespace
QPP-based with QAE (in Theorem~\ref{thm:real Renyi complexity} and~\ref{thm:integer Renyi complexity}) & $\widehat{\cO} \left(\frac{1}{\gamma^{2 - \a} \eps}\cdot \eta\right)$ & $\widehat{\cO} \left(\frac{\a\gamma + 1}{\gamma \eps}\cdot \eta\right)$ & $\widehat{\cO} \left(\frac{1}{\eps} \cdot \eta\right)$\\
\bottomrule
\end{tabular}
\caption{Comparison of algorithms on estimating quantum $\a$-R\'enyi entropies for different $\a$, in terms of the query complexity of $U_\rho$ and $U_\rho^\dagger$. Here the $\widehat{\cO}$ notation omits the $\log$ and $\a$ factors. $\gamma > 0$ is the lower bound of eigenvalues and $\kappa > 0$ is the rank of a mixed state $\rho \in \CC^{d \times d}$, and $\eps$ is the additive error of estimating $S_\a(\rho)$. QAE is short for quantum amplitude estimation.}\label{table:complete Renyi algs}
\end{table}

\end{document}